\documentclass[aps,superscriptaddress,longbibliography,twocolumn,10pt]{revtex4-2}
\usepackage[utf8]{inputenc}
\usepackage{bm}
\usepackage{times}
\usepackage{caption}
\usepackage{amsmath,amssymb}
\usepackage{tcolorbox}
\usepackage{color,xcolor}
\usepackage{relsize}
\usepackage{graphicx}
\usepackage{subfigure}
\usepackage{braket}
\usepackage{hyperref}
\usepackage[qm]{qcircuit}
\usepackage{array}
\usepackage{algorithm,algorithmic}
\usepackage{bm}
\hypersetup{colorlinks=true,citecolor=blue,linkcolor=blue,filecolor=blue,urlcolor=blue,breaklinks=true}

\usepackage{theorem}
\newtheorem{definition}{Definition}
\newtheorem{proposition}{Proposition}
\newtheorem{property}{Property}
\newtheorem{lemma}{Lemma}

\newtheorem{theorem}{Theorem}
\newtheorem{corollary}{Corollary}

\newenvironment{proof}{\noindent \textbf{{Proof~} }}{\hfill $\blacksquare$}
\newcommand{\proj}[1]{| #1\rangle\!\langle #1 |}
\newcommand{\tr}{\operatorname{Tr}}

\newcommand{\etal}{\textit{et al}.}
\newcommand{\ie}{\textit{i}.\textit{e}.}
\newcommand{\eg}{\textit{e}.\textit{g}.}

\begin{document}
\title{Near-term Efficient Quantum Algorithms for Entanglement Analysis}
\author{Ranyiliu Chen}
\affiliation{Institute for Quantum Computing, Baidu Research, Beijing 100193, China}
\affiliation{QMATH, Department of Mathematical Sciences, University of Copenhagen, Universitetsparken 5, 2100
Copenhagen, Denmark}

\author{Benchi Zhao}
\affiliation{Institute for Quantum Computing, Baidu Research, Beijing 100193, China}
\affiliation{Thrust of Artificial Intelligence, Information Hub, Hong Kong University of Science and Technology (Guangzhou), Nansha, China}

\author{Xin Wang}
\email{felixxinwang@hkust-gz.edu.cn}
\affiliation{Institute for Quantum Computing, Baidu Research, Beijing 100193, China}
\affiliation{Thrust of Artificial Intelligence, Information Hub, Hong Kong University of Science and Technology (Guangzhou), Nansha, China}
\begin{abstract}
Entanglement plays a crucial role in quantum physics and is the key resource in quantum information processing. However, entanglement detection and quantification are believed to be hard due to the operational impracticality of existing methods. This work proposes three near-term efficient algorithms exploiting the hybrid quantum-classical technique to address this difficulty. The first algorithm finds the Schmidt decomposition--a powerful tool to analyze the properties and structure of entanglement--for bipartite pure states. While the logarithm negativity can be calculated from the Schmidt decomposition, we propose the second algorithm to estimate the logarithm negativity for bipartite pure states, where the width of the parameterized quantum circuits is further reduced. Finally, we generalize our framework for mixed states, leading to our third algorithm which detects entanglement on specific families of states, and determines disdillability in general. All three algorithms share a similar framework where the optimizations are accomplished by maximizing a cost function utilizing local parameterized quantum circuits, with better hardware efficiency and practicality compared to existing methods. The experimental implementation on Quantum Leaf using the IoP CAS superconducting quantum processor exhibits the validity and practicality of our methods for analyzing and quantifying entanglement on near-term quantum devices.

\end{abstract}

\date{\today}
\maketitle

\section{Introduction}
Quantum entanglement \cite{horodecki2009quantum} is the most nonclassical manifestation of quantum mechanics. As the core ingredient in quantum information, quantum entanglement has found use in a variety of areas including quantum cryptography \cite{naik2000entangled, ekert1991quantum}, quantum chemistry \cite{mcardle2020quantum,arute2020hartree}, quantum machine learning \cite{biamonte2017quantum, cerezo2020variational}, and quantum communication \cite{bennett2020quantum, bouwmeester1997experimental, gisin2007quantum}.
	
Detecting the entanglement of a given unknown multiparty state is not simple work, let alone quantifying the entanglement. Positive (but not completely positive) map criteria, such as the well-known positive partial transpose (PPT) criterion \cite{simon2000peres}, can distinguish some entangled states by their nonpositivity after local positive maps. However, this method is less practical compared to its elegant representation from at least two points of view: the positive maps are not physically implementable, and determining positivity relies on the tomography of states \cite{Ariano2004Quantum}, which is very costly. Recently, Wang \etal \cite{wang2022detecting} proposed decomposing the nonphysical positive maps into physical Pauli operations, and estimated the positiveness using hybrid quantum-classical computation \cite{chen2021variational}. This is an effective way to detect entanglement, but the required resources in the decomposition of positive maps is still worrying.
	
The Schmidt decomposition \cite{Nielsen2010} is a fundamental and powerful tool for analyzing the entanglement in bipartite systems. According to the Schmidt decomposition, a pure bipartite state is expressed as a simpler combination of the tensor products of two orthonormal bases of subspace, that is, $\ket{\psi}_{AB} = \sum_j c_j \ket{u_j}_A \otimes \ket{v_j}_B$, where subscripts $A$ and $B$ denote the corresponding subsystems. The positive coefficient in Schmidt decomposition $c_j$ is known as the Schmidt coefficient, and the number of Schmidt coefficients is called Schmidt rank or Schmidt number, determining whether this pure state is entangled or not: if the Schmidt rank equals one, it is safe to say that there is no entanglement between the two parties. Additionally, for entangled states, entanglement measures such as Von Neumann entropy and logarithm negativity are determined by the Schmidt coefficients \cite{Vidal2002Computable}.
 
The most straightforward approach to do Schmidt decomposition on a given quantum system is state tomography \cite{Ariano2004Quantum} followed by a singular value decomposition on a classical computer. Nevertheless, it is known that state tomography consumes a huge amount of state copies, and singular value decomposition demands exponential many classical storages in the number of qubits. Alternatively, several variational quantum algorithms \cite{larose2019variational,bravo2020quantum} have been proposed to leverage the power of hybrid quantum-classical computing. In Ref. \cite{larose2019variational}, the Schmidt coefficients are read out from the variationally diagonalized marginal states. In Ref. \cite{bravo2020quantum}, local unitaries are performed and optimized on each party, and when the measurement of each party coincides the Schmidt coefficients can be evaluated.
	
In this work, we present variational algorithms for entanglement analysis that have better stability or efficiency than existing approaches. In our first algorithm for Schmidt decomposition for pure states, bilocal parameterized quantum circuits (PQCs) or quantum neural networks act on the two subsystems (inspired by Ref. \cite{bravo2020quantum,chen2021variational}), followed by a depth-2 subcircuit that calculates the fidelity with another elaborated entangled state. Von Neumann's trace inequality \cite{horn2012matrix} guarantees that when this fidelity is maximized through tuning the circuit parameters, the Schmidt decomposition is accomplished. By exploiting symmetry, we further introduce our second algorithm to estimate the logarithm negativity of bipartite pure states, where the reduction of one side of the PQC improves its efficiency mostly in classical optimization. Taking general mixed states into consideration, we propose the third algorithm for entanglement detection. We show that, it detects entanglement for certain families of mixed states of practical interest, and identifies distillability under a particular protocol. Numerical simulation and quantum device implementation in the (up to) 8-qubit and 2-qubit cases are undertaken, respectively, illustrating the practicality and validity of our algorithms in noisy intermediate-scale quantum (NISQ) devices. Our algorithms thus notably provide efficient and practical ways to analyze and quantify entanglement on near-term quantum devices.
	
The rest of the paper is organized as follows: Section \ref{sec:notations} gives notation and preliminaries used in this paper. Section \ref{sec:theoframe} provides the theoretic foundations of our algorithm design of which the details about cost function evaluation and the adopted optimization are introduced in Section \ref{sec:algo}. In Section \ref{sec:experiments} we exhibit experiments including numerical simulation and quantum device implementation. Finally, Section \ref{sec:conclusion} concludes this work and discusses outlooks.

\section{Notation and preliminaries}\label{sec:notations}

\subsection{Notation}
    In this paper, the standard computational basis vectors of an $n$-qubit system are denoted by $\ket{j}$, $j=0,\dots,2^n-1$, with matrix representations $\ket{0}=(1,0,\dots,0)^\intercal,\ket{1}=(0,1,\dots,0)^\intercal,\dots,\ket{2^n-1}=(0,0,\dots,1)^\intercal$. For composed systems, we write an operator with a subscript indicating the system that the operator acts on, such as $M_{AB}$, and write $M_A := \operatorname{Tr}_B M_{AB}$. We mainly focus on the natural partition of the $2n$-qubit system, which corresponds to a Hilbert space $\mathcal{H}_{AB}=\mathcal{C}^d\otimes\mathcal{C}^d$ where $d=2^n$ is the local dimension.
    
\subsection{Schmidt decomposition for pure states}
    The Schmidt decomposition is a fundamental tool to characterize the entanglement of bipartite pure states. Given a pure state $\ket{\psi}_{AB}$ living in the composed Hilbert space $\mathcal{H}_{AB}$, it always admits the Schmidt decomposition \cite{Nielsen2010}:
	\begin{align*}
		\ket{\psi}_{AB}=\sum_{j=0}^{R-1}c_j\ket{u_j}_A\ket{v_j}_B,
	\end{align*}
    where $\{\ket{u_j}_A\}$ and $\{\ket{v_j}_B\}$ are some orthonormal bases in marginal spaces $\mathcal{H}_A$ and $\mathcal{H}_B$, respectively. Here the \textit{Schmidt coefficients} $c_j>0$ are ordered decreasingly, \ie, $c_j\ge c_k$ for all $j<k$. The number of positive coefficients $R$ is called the \textit{Schmidt rank}, satisfying $R\le\min\{d_A,d_B\}$, where $d_A$ and $d_B$ are the dimensions of $\mathcal{H}_A$ and $\mathcal{H}_B$, respectively.
	
	Once given the Schmidt decomposition, one can deduce entanglement measures from these {Schmidt coefficients}. For example, the logarithm negativity follows as $E_N(\ket{\psi}_{AB})=\log_2N(\ket{\psi}_{AB})$, where $N(\ket{\psi}_{AB})=1/2[(\sum_jc_j)^2-1]$ is the negativity \cite{Vidal2002Computable}.
    
\subsection{Reduction map as a separability criterion}
    The reduction map \cite{cerf1999reduction} $\mathcal{R}$ defined by
    \begin{align*}
        \mathcal{R}(X)=\tr X\cdot I-X
    \end{align*}
    is a positive but not completely positive (PnCP) map. Like every PnCP map, the reduction map $\mathcal{R}$ can be used as a separability criterion, in the sense that any separable state $\rho_{AB}$ satisfies $(\mathcal{R}_A\otimes\operatorname{id}_B)(\rho_{AB})=I\otimes\rho_B-\rho_{AB}\ge0$. The set of bipartite states that are positive under the partial reduction map is called the $\mathcal{R}$ state; then the set of separable states, denoted by $\operatorname{SEP}$, is a subset of the $\mathcal{R}$-state. It is worth mentioning that PPT state is a subset of the $\mathcal{R}$ state, and when $d_Ad_B\le6$ they both become necessary conditions for entanglement \cite{cerf1999reduction}. Also, any state violating the reduction criterion is distillable \cite{horodecki1998reduction}.
    
    \subsection{Maximally entangled states}
    In a composed system $\mathcal{H}_{AB}=\mathcal{C}^d\otimes\mathcal{C}^d$, the maximally entangled states are the states with Schmidt rank $R=d$ and Schmidt coefficients $c_j=1/\sqrt{d}$, \ie, states of the form 
    \begin{align*}
    \ket{\Phi}_{AB}=\frac{1}{\sqrt{d}}\sum_{j=0}^{d-1}\ket{u_j}_A\ket{v_j}_B,
    \end{align*}
    where $\ket{u_j}_A$ and $\ket{v_j}_B$ are some orthonormal bases in $\mathcal{H}_{A}$ and $\mathcal{H}_{B}$, respectively. We denote the set of all maximally entangled states by $\operatorname{MAXE}$:
    \begin{align*}
        \operatorname{MAXE}:=\{\ket{\Phi}|\ket{\Phi}=\frac{1}{\sqrt{d}}\sum_{j=0}^{d-1}\ket{u_j}_A\ket{v_j}_B\}.
    \end{align*}

    It is clear that the maximally entangled states of the same dimension are only different in local orthonormal bases. Among states in $\operatorname{MAXE}$ we denote by $\ket{\Phi^+}_{AB}=1/\sqrt{d}\sum_j\ket{j}_A\ket{j}_B\in\operatorname{MAXE}$ the maximally entangled state in the computational basis. A frequently used identity in entanglement theory about $\ket{\Phi^+}$ is the ``transpose trick'': for any operator $M$, 
    \begin{align}
        (M_A\otimes I_B)\ket{\Phi^+}_{AB}=(I_A\otimes M_B^{\intercal})\ket{\Phi^+}_{AB},
        \label{transpose}
    \end{align}
    where $M^{\intercal}$ denotes the transpose (with respect to the computational basis) of $M$.
    
    We remark that $\operatorname{MAXE}$ can also be expressed as
    \begin{align*}
        \operatorname{MAXE}=\{\ket{\Phi}|\ket{\Phi}=U_A\otimes V_B\ket{\Phi^+}\}
    \end{align*}
    where $U_A,V_B\in\mathcal{U}_d$ are unitaries on $\mathcal{H}_{A}$ and $\mathcal{H}_{B}$, respectively. Note that we can further reduce the unitary operator on $\mathcal{H}_{B}$ via the transpose trick: 
    \begin{align*}
        \operatorname{MAXE}=&\{\ket{\Phi}|\ket{\Phi}=V_B^\intercal U_A\otimes I_B\ket{\Phi^+}\}\\
        =&\{\ket{\Phi}|\ket{\Phi}=U_A\otimes I_B\ket{\Phi^+}\},
    \end{align*}
    since $V_B^\intercal U_A\in\mathcal{U}_d$. 

\section{Theoretical framework}\label{sec:theoframe}

In this section, we provide the theoretical foundation for our variational algorithms. We start with the pure-state case whose Schmidt decomposition is sufficient to characterize its entanglement properties. Then we extend our method by looking at general mixed states and building the connection with the reduction separability criterion.

\subsection{Schmidt decomposition for bipartite pure states}

Consider a bipartite state $\ket{\psi}_{AB}$ operated by \textit{local} unitaries $U_A\otimes V_B$. This unitary operator will not affect any entanglement properties, \eg, the Schmidt coefficients. Our key observation is that, the maximal fidelity (maximized over the unitaries) between $U_A\otimes V_B\ket{\psi}_{AB}$ and some elaborated entangled state reflects the weighted sum of the Schmidt coefficients of $\ket{\psi}_{AB}$. It is rigorously stated in the following theorem: 

\begin{theorem}
	Let $U_A$ and $V_B$ be unitaries on systems $A$ and $B$, respectively. Let $\ket{\Psi}_{AB}=\sum_jp_j\ket{j}_A\ket{j}_B$ where $\{p_j\}$ is positive and ordered \emph{strictly decreasingly}. For any bipartite state $\ket{\psi}_{AB}=\sum_{j}c_j\ket{u_j}_A\ket{v_j}_B$ with decreasing coefficients $c_j$, denote $\ket{\widetilde{\psi}}_{AB}=U_A\otimes V_B\ket{\psi}_{AB}$, it holds that
	\begin{align}
		\max_{U_A,V_B}F(\ket{\widetilde{\psi}}_{AB},\ket{\Psi}_{AB})={\sum_jp_jc_j},
		\label{eq:real}
	\end{align}
	and when the maximal of Eq. \eqref{eq:real} is reached, the Schmidt coefficients can be readout as 
	\begin{align}
	c_j=F({\widetilde{\rho}}_{A},\ket{j}_A),
	\label{eq:valuereadout}
	\end{align}
	where $\widetilde{\rho}_{A}=\operatorname{Tr}_B[\ket{\widetilde{\psi}}\!\bra{\widetilde{\psi}}_{AB}]$.
	Furthermore if $c_j$ is not degenerate, \ie, $c_{j-1}>c_j>c_{j+1}$, the orthonormal vectors can be prepared (up to global phases $e^{i\theta_{A,B}^{(j)}}$) by
	\begin{align}&e^{i\theta_A^{(j)}}\ket{u_j}_A=U_A^\dagger\ket{j}_A,\nonumber\\ &e^{i\theta_B^{(j)}}\ket{v_j}_B=V_B^\dagger\ket{j}_B,
	\label{eq:vectorprepare}
	\end{align}
	for some real $\theta_{A,B}^{(j)}$.
	\label{thm:real}
\end{theorem}

\begin{proof}
Recall the operator-vector correspondence mapping \cite{watrous_2018} defined by the action on base vectors:
\begin{align*}
    \operatorname{vec}(\ket{j}\!\bra{k})=\ket{j}\ket{k}.
\end{align*}

Since the $\operatorname{vec}$ mapping is bijective and isometric, we define its inverse mapping as
\begin{align*}
    \operatorname{mat}(\ket{j}\ket{k}):=\operatorname{vec}^{-1}(\ket{j}\ket{k})=\ket{j}\!\bra{k},
\end{align*}
which is also isometric. Let $X=\operatorname{mat}(\ket{\widetilde{\psi}}_{AB})$ and $Y=\operatorname{mat}(\ket{\Psi}_{AB})$; then $F(\ket{\widetilde{\psi}}_{AB},\ket{\Psi}_{AB})=|\braket{\Psi|\widetilde{\psi}}_{AB}|=|\tr X^\dagger Y|$ by isometry. What is more, $c_j$ and $p_j$ are singular values of $X$ and $Y$ in decreasing orders, respectively, bringing about $|\tr X^\dagger Y|\le\sum_jc_jp_j$ via Von Neumann's trace inequality \cite{horn2012matrix}. To sum up, we have
\begin{align}
    &F(\ket{\widetilde{\psi}}_{AB},\ket{\Psi}_{AB})\nonumber\\
    =&|\bra{\Psi}_{AB}U_A\otimes V_B\ket{\psi}_{AB}|\nonumber\\
    =&|\tr X^\dagger Y|\nonumber\\
    \le&\sum_jc_jp_j.
    \label{eq:orderedsum}
\end{align}
Thus the fidelity will not exceed $\sum_jc_jp_j$. Evidently, when $U_A$ maps $\{\ket{u_j}_A\}$ to $\{\ket{j}_A\}$ and $V_B$ maps $\{\ket{v_j}_B\}$ to $\{\ket{j}_B\}$, the maximal is reached. Then Eq. \eqref{eq:real} holds.

When the maximal of Eq. \eqref{eq:real} is reached, by Von Neumann's trace inequality $U_A\ket{u_j}_A$ (correspondingly, $V_B\ket{v_j}_B$) falls in the subspace spanned by $\{\ket{k}_A|c_k=c_j\}$ (correspondingly, $\{\ket{k}_B|c_k=c_j\}$); thus Eq. \eqref{eq:valuereadout} holds; if $c_j$ is further nondegenerate, then Eq. \eqref{eq:vectorprepare} holds naturally, which completes the proof.
\end{proof}

\textbf{Remark:} Inequality \eqref{eq:orderedsum} was also proved in Refs. [\cite{zhang2016evaluation}, Eq. 5] and [\cite{antipin2020lower}, Eq. 23], and from the von Neumann's trace inequality as well. Beyond the inequality itself, Theorem \ref{thm:real} also shows that, with nondegenerated Schmidt coefficients, the Schmidt vectors can be prepared (up to phases) when the inequality saturates, which is part of the Schmidt decomposition task.

Theorem \ref{thm:real} enables the design of our variational algorithm. The idea is straightforward: we use parameterized quantum circuits to implement $U_A,V_B$ respectively. Then we set the cost function to be some monotonic function of the fidelity $F(\ket{\widetilde{\psi}}_{AB},\ket{\Psi}_{AB})$ (in Sec. \ref{subsec:costfunction} we choose $F^2$ to be the cost function for convenience). By tuning the circuits' parameters we can maximize the cost function. Finally after the optimization we can readout the required Schmidt coefficient. In Section \ref{subsec:costfunction} the evaluation of the cost function and the readout process will be discussed in detail. It is worth noting that, even though we can prepare each decomposition vector $\ket{e_j}_{A,B}$, we might not be able to reconstruct $\ket{\psi}_{AB}$ via $U_A$ and $V_B$ because the relative phases for each component might not be equal, \ie, $\theta^{(j)}_{A,B}$ are not always the same for different $j$.

\subsection{Logarithm negativity for bipartite pure states}

Entanglement measures quantify the entanglement between quantum systems, being a key figure of merit in entanglement distillation and many other protocols. Known methods, however, are not directly applicable in most cases on near-term devices. In this part, we introduce the variational algorithm to estimate the logarithm negativity of pure bipartite states, which consume fewer resources. After extracting all the Schmidt coefficients from the above approach, entanglement measures such as logarithm negativity can be calculated by (classical) postprocessing \cite{Vidal2002Computable}. Significantly, in the case where one would like to estimate the logarithm negativity directly and have no interest in each coefficient, we can provide a simpler variational estimation by substituting $\ket{\Psi}_{AB}$ with the bipartite maximally entangled state and using the transpose trick. Specifically, we have the following corollary.
\begin{corollary}
	Let $U_A$ be a unitary operator on system A. For any bipartite state $\ket{\psi}_{AB}=\sum_{j}c_j\ket{u_j}_A\ket{v_j}_B$ with decreasing coefficients $c_j$, denote $\ket{\widetilde{\psi}}_{AB}=U_A\otimes I_B\ket{\psi}_{AB}$, it holds that
	\begin{align}
		&\max_{U_A}F(\ket{\widetilde{\psi}}_{AB},\ket{\Phi^+}_{AB})\nonumber\\
  =&\max_{\ket{\Phi}\in\operatorname{MAXE}}F(\ket{{\psi}}_{AB},\ket{\Phi})\nonumber\\
  =&\frac{\sum_jc_j}{\sqrt{d}},
		\label{eq:norm}
	\end{align}
	where $\operatorname{MAXE}$ is the set of all maximally entangled states.
	\label{cor:norm}
\end{corollary}

\begin{proof}
	In Eq. \eqref{eq:real}, let $p_j=1/\sqrt{d}$, we have
	\begin{align*}
		F(U_A\otimes V_B\ket{{\psi}}_{AB},\ket{\Phi}_{AB})\le\sum_jc_jp_j=\frac{\sum_jc_j}{\sqrt{d}}.
	\end{align*}
	Thus the maximal fidelity will not exceed $\frac{\sum_jc_j}{\sqrt{d}}$. Evidently, when $U_A$ maps $\{\ket{u_j}_A\}$ to $\{\ket{j}_A\}$ and $V_B$ maps $\{\ket{v_j}_B\}$ to $\{\ket{j}_B\}$, the maximal is reached. 
	
	Now note that
	\begin{align*}
		&\max_{U_A,V_B}F(U_A\otimes V_B\ket{{\psi}}_{AB},\ket{\Phi^+}_{AB})\nonumber\\
		=&\max_{U_A,V_B}|\bra{{\psi}}_{AB}U_A^{\dagger}\otimes V_B^{\dagger}\ket{\Phi^+}_{AB}|\nonumber\\
		=&\max_{U_A,V_B}|\bra{{\psi}}_{AB}U_A^{\dagger}V_B^*\otimes I_B\ket{\Phi^+}_{AB}|\\
		=&\max_{U_A}F(U_A\otimes I_B\ket{{\psi}}_{AB},\ket{\Phi^+}_{AB})\nonumber\\
		=&\max_{\ket{\Phi}\in\operatorname{MAXE}}F(\ket{\psi}_{AB},\ket{\Phi}_{AB})\nonumber,
	\end{align*}
	where the second equation follows from the transpose trick (Eq. \eqref{transpose}). Then Eq. \eqref{eq:norm} holds, completing the proof.
\end{proof}

(A Similar result is also discussed in Ref. \cite{Zhao_2010}.) Corollary \ref{cor:norm} gives an estimation of the $L1$ norm (the sum) of the Schmidt coefficients, reflecting entanglement properties of the state. For example, note that the logarithm negativity of $\ket{\psi}_{AB}$ can be written as \cite{Vidal2002Computable}
\begin{align*}
E_N(\ket{\psi}_{AB})=\log_2(2N+1)=\log_2(\sum_jc_j)^2,
\end{align*}
where $N(\ket{\psi}_{AB})=1/2[(\sum_jc_j)^2-1]$ is the negativity of $\ket{\psi}_{AB}$.  Specifically, since Eq. \eqref{eq:norm} is an optimization over unitaries, a variational quantum algorithm implements the optimization via tuning parameterized quantum circuits on near-term devices. The fidelity in Eq. \eqref{eq:norm} then is obtained by measurement in the computational basis after the subcircuit whose inverse prepares $\ket{\Phi^+}$ from $\ket{0^n}$. In this way, entanglement measures such as logarithm negativity can be inferred without tomography of the state. (In Section \ref{sec:algo} more details on the variational quantum algorithm will be provided.) Notably, Corollary \ref{cor:norm} enables us to employ a one-side parameterized quantum circuit (on either the $A$ or $B$ system) in the variational learning. This reduces the parameters required in the parameterized quantum circuit by at least half, saving an appreciable amount of computational resources in classical optimization.

\subsection{Entanglement detection for the general bipartite state}

For pure states, entanglement can be easily detected by measuring the purity of the reduced quantum state. However, for mixed states, detecting entanglement becomes more challenging. While the purity and some related entropies (like Tsallis-2 entropy $1-\tr[\rho^2]$ and second-order Renyi entropy $-\log\tr[\rho^2]$) can be computed efficiently, estimating entropies in general is difficult. Moreover, purity and these specific entropies typically do not provide much operational information, such as distillability of the state. Wang \etal \cite{wang2022detecting}, realized the positive map criteria by decomposing the maps into Pauli operators, but this decomposition can consume exponentially many resources in multiqubit cases.

In the above logarithm negativity estimation for pure states we considered the maximal fidelity between a pure state and the maximally entangled states. It should not be surprising that this maximal fidelity also reflects some entanglement properties for general mixed states. In fact, the maximal overlap between a mixed state and the family of maximally entangled states is called the \emph{fully entangled fraction} \cite{Zhao_2010}.

\begin{definition}
The \emph{fully entangled fraction} $\chi$ of any bipartite state $\rho_{AB}$ is defined as the maximal state overlap to $\operatorname{MAXE}$:
\begin{align*}
    \chi(\rho_{AB}):=&\max_{\ket{\Phi}\in\operatorname{MAXE}}F^2(\rho_{AB},\ket{\Phi}_{AB})\nonumber\\
    =&\max_{\ket{\Phi}\in\operatorname{MAXE}}\tr\left[\rho_{AB}\proj{\Phi}\right]\nonumber\\
    =&\max_{U_A\in\mathcal{U}_d}\tr\left[\rho_{AB}U_A\otimes I_B\proj{\Phi^+}U_A^\dagger\otimes I_B\right],
\end{align*}
with $\operatorname{MAXE}$ the set of all maximally entangled states. 

A \emph{$\chi$ state} is defined as the set of all bipartite states $\rho_{AB}$ satisfying $\chi(\rho_{AB})\le1/d$, where $d$ is the local dimension.
\end{definition}

The fully entangled fraction $\chi$ was previously studied in the qubit-qubit case in Refs. \cite{Local2000Badzia,Fidelity2002Verstraete}, and more generally, in Refs.  \cite{Zhao_2010,Ganguly2011entanglement,vidal2000approximate}. Specifically, Ganguly \etal \cite{Ganguly2011entanglement} showed that the $\chi$-state set is convex and compact.

We point out that $\chi\le1/d$ is a valid separability criterion: recall that the $\mathcal{R}$ state is the set of states that remain positive under the partial $\mathcal{R}$ map. We illustrate that the $\chi$ state contains the $\mathcal{R}$ state.
\begin{lemma}
\label{lemma:containingRstate}
The $\mathcal{R}$ state is a subset of the $\chi$ state.
\end{lemma}

\begin{proof}
    By definition, we need to show that $\chi(\rho_{AB})\le1/d$ for any $\rho_{AB}\in\mathcal{R}$ state.
    
    Given any $\mathcal{R}$ state $\rho_{AB}$, it holds that
    \begin{align*}
        &\mathcal{R}_A\otimes\operatorname{id}_B(\rho_{AB})\ge0\\
        \Leftrightarrow& U_A\circ\mathcal{R}_A\otimes\operatorname{id}_B(\rho_{AB})\ge0\\
        \Rightarrow& d\tr[U_A\circ\mathcal{R}_A\otimes\operatorname{id}_B(\rho_{AB})\proj{{\Phi^+}}]\ge0,
    \end{align*}
    for any unitary map $U_A$. The right-hand side satisfies
    \begin{align*}
        &d\tr[U_A\circ\mathcal{R}_A\otimes\operatorname{id}_B(\rho_{AB})\proj{{\Phi^+}}]\nonumber\\
        =&d\tr[\rho_{AB}\mathcal{R}_A^\dagger\circ U_A^\dagger\otimes\operatorname{id}_B(\proj{{\Phi^+}})]\nonumber\\
        =&d\tr[\rho_{AB}\mathcal{R}_A\circ U_A^\dagger\otimes\operatorname{id}_B(\proj{{\Phi^+}})]\nonumber\\
        =&1-\tr\left[\rho_{AB}U_A^\dagger\otimes I_B\proj{\Phi^+}U_A\otimes I_B\right]\ge0,
    \end{align*}
    where in the first equivalence $\mathcal{R}_A^\dagger$ is the adjoint map of $\mathcal{R}_A$, and in the second equivalence we used $\mathcal{R}_A^\dagger=\mathcal{R}_A$. Then
    \begin{align*}
        \forall U_A,\tr\left[\rho_{AB}U_A\otimes I_B\proj{\Phi^+}U_A^\dagger\otimes I_B\right]\le1/d,
    \end{align*}
    leading to
    \begin{align*}
        \chi(\rho_{AB})=&\max_{U_A\in\mathcal{U}_d}\tr\left[\rho_{AB}U_A\otimes I_B\proj{\Phi^+}U_A^\dagger\otimes I_B\right]\\
        \le&1/d,
    \end{align*}
    which completes the proof.
\end{proof}

With Lemma \ref{lemma:containingRstate} one can think of $\chi$ as a separability criterion that is no stronger than the reduction criterion. That is, all states with fully entangled fraction $\chi(\rho)>1/d$ will also violate the reduction criterion. It is known that states with fully entangled fraction $\chi(\rho)>1/d$ can be distilled by the protocol consisting of twirling and the generalized XOR operation \cite{horodecki1998reduction}. Therefore, our criterion determines distillability for this protocol.

Geometrically, the $\mathcal{R}$ state is a convex set that contains SEP and PPT. Consider the map $U_A\circ\mathcal{R}_A\otimes\operatorname{id}_B$ for some fixed unitary $U_A$, and let the $\Tilde{\mathcal{R}}$ state be the set of states that is positive under this map. Then the $\Tilde{\mathcal{R}}$ state is equal to the ${\mathcal{R}}$ state since unitary operator preserves positivity. Then $\tr[W\rho]=0$ corresponds to the hyperplane `outside' of the $\Tilde{\mathcal{R}}$ state, where $W=\mathcal{J}(U_A\circ\mathcal{R}_A\otimes\operatorname{id}_B)=\mathcal{R}_A\circ U_A^\dagger\otimes\operatorname{id}_B(\proj{{\Phi^+}})$ is the Choi operator of the partial $\Tilde{R}$ map. Take $U_A$ over the unitary group; those hyperplanes envelop a convex set, which is essentially the $\chi$ state by our definition (see Fig. \ref{fig:geometry}).

\begin{figure*}[t]
	\centering
	\includegraphics[width=0.6\textwidth]{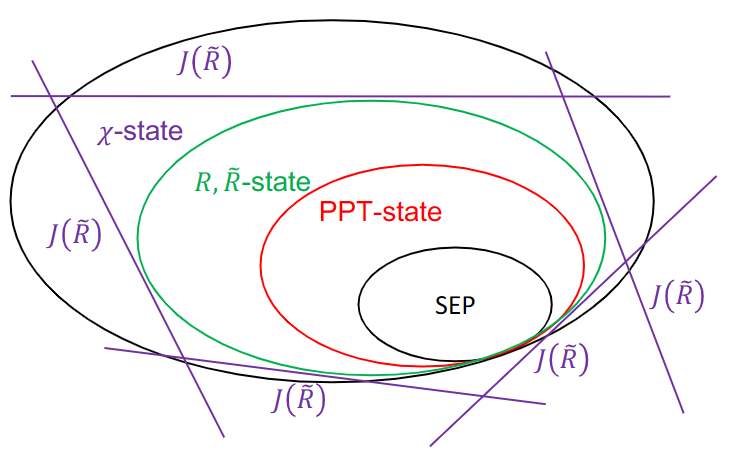}
	\caption{The geometrical explanation of the $\chi$ state. The black, red, and green circles correspond to the separable, PPT, and $\mathcal{R}$ state, respectively. The purple lines are hyperplanes defined by $\tr[\mathcal{J}(\tilde{R})\rho]=0$, where $\mathcal{J}(\tilde{R})$ is the Choi operator of the local $\tilde{\mathcal{R}}$ map. Note that each line corresponds to a different $\tilde{\mathcal{R}}$ map defined by different $U_A$.}
	\label{fig:geometry}
\end{figure*}

We study the value of $\chi$ for certain families of entangled mixed states that are of practical interest. We have the following result.

\begin{proposition}
The $\chi(\rho)>1$ criterion can detect entanglement on the following families of states
\begin{itemize}
    \item isotropic states: $\rho=p\proj{\Phi^+}+(1-p)I/d^2$, $p\in[\-\frac{1}{d^2-1},1]$;
    \item S states: $\rho=p\proj{\Phi^+}+(1-p)\proj{00}$, $p\in[0,1]$;
    \item Werner states ($2\times2$): $\rho=\frac{1}{d^3-d}((d-\alpha)I+(1-d\alpha)F)$, where $d=2$, $\alpha\in[-1,1]$, and $F$ is the flip operator;
    \item Bit-and-phase-flipped Bell states $\rho=p\rho_1+(1-p)X\otimes I\rho_1X\otimes I$, $\rho_1=q\proj{\Phi^+_2}+(1-q)Z\otimes I\proj{\Phi^+_2}Z\otimes I$.
\end{itemize}
\label{prop:4cases}
\end{proposition}

\begin{proof}
    We calculate $\chi$ for each case individually.
    \begin{itemize}
        \item For isotropic states, 
        $$
        \chi(\rho)=p\chi(\proj{\Phi^+})+(1-p)/d^2=p+(1-p)/d^2.
        $$
        So $\chi(\rho)>1/d$ if and only if $p>\frac{1}{d+1}$. On the other hand, it is known \cite{horodecki1998reduction} that isotropic states are entangled if and only if $p>\frac{1}{d+1}$, and the reduction criterion detects all entangled isotropic states. So $\chi$ also detects all entangled isotropic states.
        \item For S states,
        $$
        \chi(\rho)\le p\chi(\proj{\Phi^+})+(1-p)\chi(\proj{00})=p+(1-p)/d.
        $$
        This bound is achievable by taking $U=I$. So $\chi(\rho)=p+(1-p)/d$. For all $p>0$, $\chi(\rho)>p/d+(1-p)/d=1/d$. On the other hand, all $S$ states are entangled unless $p=0$, as they all violate the PPT criterion.
        \item For Werner states, for $d>2$, it is known that the reduction criterion does not detect any entanglement for Werner states \cite{horodecki1998reduction}. So the best hope for is the 2-dimensional case. For $d=2$, Werner states are equivalent to isotropic states up to local unitary operator, and $\chi$ is local-unitary invariant. So $\chi$ detects all entangled 2-dimensional Werner states.
        \item For bit-and-phase-flipped Bell states, decompose the optimizing unitary operator as $U=R_z(z_1/2)R_y(y/2)R_z(z_2/2)$; then direct calculation shows that
        \begin{align*}
        \chi(\rho)=&\max_{z_1,y,z_2}\cos^2(y)(1+(1-2p)(1-2q)\cos(z_1+z_2))/2\\
        =&1/2+|(1-2p)(1-2q)|/2\ge1/2.
        \end{align*}
        It is known that these states are entangled unless $p=q=1/2$. So $\chi$ detects all such entangled states.
    \end{itemize}
    This completes the proof.
\end{proof}

Proposition \ref{prop:4cases} shows that our criterion detects entanglement for states that are relevant in practical applications.

One might also hope that $\chi$ forms some entanglement measure \cite{bengtsson2017geometry}. Unfortunately, there are counterexamples against it.

\begin{property}
\label{property:LOCC}
It holds that $\chi$ is local-unitary invariant, but not local operation and classical communication (LOCC) nonincreasing.
\end{property}

\begin{proof}
It is evident that $\chi$ is invariant under local unitary operator since the optimization is over all local unitaries.

We show that $\chi$ is not LOCC nonincreasing by contradiction. Suppose that $\chi$ is LOCC nonincreasing, then $\chi$ must be constant on separable states since separable states are mutually LOCC convertible. Then consider $\rho_0=I_{d^2}/d^2$ the maximally mixed state and $\rho_1=\proj{00}$ the pure separable state, we have
\begin{align*}
    \chi(\rho_0)=\frac{1}{d^2}\tr[\proj{\Phi^+}]=\frac{1}{d^2},
\end{align*}
while $\chi(\rho_1)=(2N(\ket{00})+1)/d=1/d\neq\chi(\rho_0)$ from Property \ref{cor:norm}. This leads to a contradiction, so $\chi$ is not LOCC nonincreasing.
\end{proof}

Also, an example of the $\chi$ state not equaling the R state is provided in Appendix \ref{sec:SM_apnoisy}.

\section{Variational Quantum Algorithms}\label{sec:algo}
With the theoretical framework in hand, we develop the variational quantum algorithms for Schmidt decomposition, logarithm negativity estimation, and entanglement detection. Here, the variational quantum algorithm (VQA) \cite{Cerezo2020,Bharti2021,Endo2020} is a popular paradigm for near-term quantum applications, which uses a classical optimizer to train parameterized quantum circuits to achieve certain tasks. VQA is applied to solve problems in many areas, including ground and excited states preparations \cite{Cao2019,Nakanishi2019,Higgott2018}, quantum data compression \cite{Romero2017,Cao2020,Wang2020d}, combinatorial optimization \cite{Farhi2014}, quantum classificationr~\cite{Li2021-classifier,Schuld2018,Li2021}, and quantum metrology \cite{Meyer2020,Koczor2019a}

In this section, we present the details of cost function evaluation and optimization methods in the algorithm design. The diagram of our algorithm is shown in Fig. \ref{fig:diagram} and the algorithm boxes are given in Appendix \ref{sec:SM_algo}.

\begin{figure*}[t]
	\centering
	\includegraphics[width=0.9\textwidth]{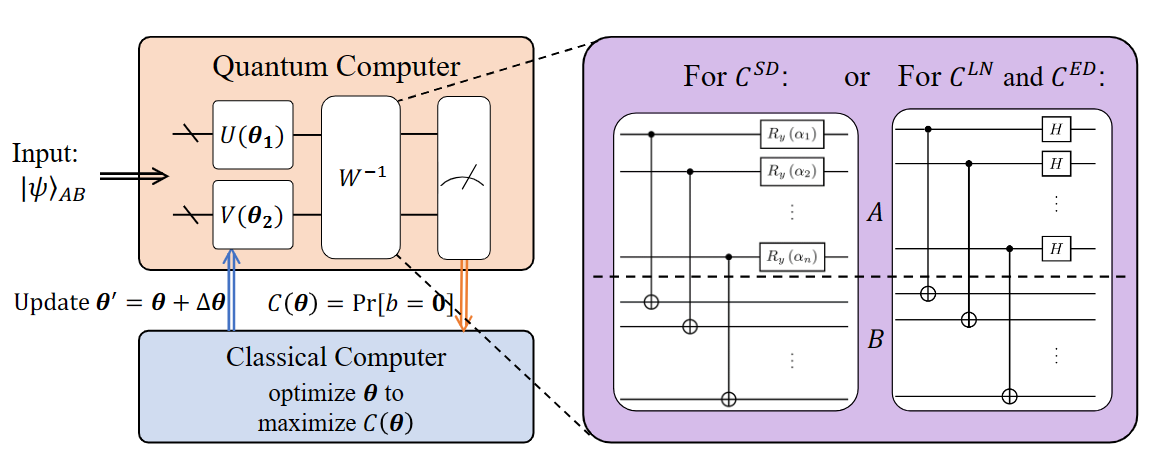}
	\caption{Diagram of the variational quantum algorithms for Schmidt decomposition and logarithm negativity. Apply parameterized quantum circuits $U(\bm{\theta_1})$ and $V(\bm{\theta_2})$ onto the two parts $A$ and $B$ respectively [for logarithm negativity estimation and entanglement detection $I_B$ can replace $V(\bm{\theta_2})$]. Then operate an inverse of the state-preparing circuit ($W$ for Schmidt decomposition and $W'$ for logarithm negativity estimation and entanglement detection). After measurement, we take the probability of the all-zero measurement outcome as the cost function, and the classical optimizer is used to maximize the cost function by updating the parameters in the PQC iteratively.}
	\label{fig:diagram}
\end{figure*}

\subsection{Cost functions}\label{subsec:costfunction}

We first look at the cost function in the Schmidt decomposition task. Recall that $\max F(\ket{\widetilde{\psi}}_{AB},\ket{\Psi}_{AB})={\sum_jp_jc_j}$ according to Theorem \ref{thm:real}. In the variational algorithm, we set our cost function to be:
\begin{equation}
    C^{SD}:=F^2(\ket{\widetilde{\psi}}_{AB},\ket{\Psi}_{AB})=\operatorname{Tr}\left[\widetilde{\rho}_{AB}\ket{\Psi}\!\bra{\Psi}_{AB}\right].
\label{eq:C^SD}
\end{equation}

Recall that $\widetilde{\rho}_{AB}=\ket{\widetilde{\psi}}\!\bra{\widetilde{\psi}}_{AB}$ and $\ket{\widetilde{\psi}}_{AB}=U_A\otimes V_B\ket{\psi}_{AB}$. Evidently, $F$ reaches its maximal if and only if $C^{SD}$ reaches the maximal. Thus, after the cost function $C$ is optimized the Schmidt coefficients can be read out through measurement:
\begin{align*}
    c_j=F(\ket{\widetilde{\psi}}_{A},\ket{j}_{A})=\sqrt{\operatorname{Tr}\left[\widetilde{\rho}_{A}\ket{j}\!\bra{j}_{A}\right]}.
\end{align*}

We then show that the cost function can be evaluated efficiently on a NISQ device. Suppose that circuit $W$ prepares the state $\ket{\Psi}_{AB}=W\ket{0}_{AB}$. By the cyclic property of trace we have
\begin{align}
    C^{SD}=&\operatorname{Tr}[\widetilde{\rho}_{AB}W\ket{0}\!\bra{0}_{AB}W^\dagger]\nonumber\\
    =&\operatorname{Tr}[W^\dagger\widetilde{\rho}_{AB}W\ket{0}\!\bra{0}_{AB}].
\label{eq:C^SD_W}
\end{align}

Thus, the cost function $C^{SD}$ is equal to the probability of all-zero measurement outcomes of $\ket{\widetilde{\psi}}_{AB}$ acted on by $W^\dagger=W^{-1}$. So the key to evaluating $C$ is to implement $W$ ($W^{-1}$) efficiently. 

Note that $W$ can be constructed by $W_A$ which prepares the superposition $\sum_jp_j\ket{j}_A$ in system $A$ and $n$ CNOT gates connecting all qubit pairs across systems $A$ and $B$. If we preset the coefficients $p_j$ then constructing the circuit for $W_A$ is essentially the amplitude encoding \cite{schuld2021supervised} task, for which the best-known construction\cite{Plesch_2011} requires $O(2^n)$ CNOT count and $O(2^n)$ depth. Nevertheless, in our algorithm, we only require coefficients $\{p_j\}$ to be decreasing. We provide a circuit construction for this `weak' state preparation task that uses $n$ single-qubit gates in parallel. In fact, the circuit composed of $n$ paralleled $y$-axis rotation gates
\[ \Qcircuit @C=1em @R=0.7em {
	&\lstick{\ket{0}}&\gate{R_y(\alpha_{1})}&\qw\\
	&\lstick{\ket{0}}&\gate{R_y(\alpha_{2})}&\qw&\rstick{\sum_jp_j\ket{j}_A}\\
	&...\\
	&\lstick{\ket{0}}&\gate{R_y(\alpha_{n})}&\qw\\
	{\gategroup{1}{4}{4}{4}{1.75em}{\}}}
}
\]
with carefully chosen parameters $\{\alpha_j\}$, satisfies the condition that $p_j$ decreases for arbitrary $n$. When $n$ trivially equals $1$, the parameter $\alpha_1$ can be set to $0$. We refer the reader to Appendix \ref{sec:SM_alpha} for the selection of $\alpha_j$ for $n\ge2$. Note that the total depth in this construction of $W$ is 2 with CNOT count $n$, leading to an efficient implementation of $W$ (thus $W^{-1}$) in the cost function evaluation.

We next turn to the logarithm negativity estimation. In this case, since $\max F(\ket{\widetilde{\psi}}_{AB},\ket{\Phi^+}_{AB})=\sum_jc_j/\sqrt{d}$ according to Corollary \ref{cor:norm}, the cost function is set as:
\begin{align*}
    C^{LN}:=F^2(\ket{\widetilde{\psi}}_{AB},\ket{\Phi^+}_{AB})=\operatorname{Tr}\left[\widetilde{\rho}_{AB}\proj{\Phi^+}_{AB}\right].
\end{align*}

Evidently, $C^{LN}_{\max}=(\sum_jc_j)^2/d$. We use the variational method to maximize the cost function $C^{LN}$, and then the logarithm negativity is estimated as
\begin{align}\label{eq:LN_cal}
E_N(\ket{\psi}_{AB})=\log_2(C^{LN}_{\max}d)=\log_2(C^{LN}_{\max})+n.
\end{align}

The cost function $C^{LN}$ also has an efficient evaluation using a similar technique to that discussed above. Note that the maximally entangled state $\ket{\Phi}_{AB}$ can be prepared by a depth-2 circuit $W'$ consisting of $n$ Hadamard gates and $n$ CNOT gates. With the inverse circuit $W'^\dagger$ in hand, the cost function $C^{LN}$ is evaluated as
\begin{align}
    C^{LN}=&\operatorname{Tr}\left[\widetilde{\rho}_{AB}\proj{\Phi^+}_{AB}\right]\nonumber\\
    =&\operatorname{Tr}\left[W'^\dagger\widetilde{\rho}_{AB}W'\ket{0}\!\bra{0}_{AB}\right],
        \label{eq:LN_loss_func}
    \end{align}

which is the probability of the all-zero measurement outcome of $\ket{\widetilde{\psi}}_{AB}$ acted on by $W'^\dagger$.

Finally, for the entanglement detection algorithm, we use the same cost function as we used when estimating logarithm negativity:
    \begin{align}
        C^{ED}=C^{LN}=&\operatorname{Tr}\left[\widetilde{\rho}_{AB}\proj{\Phi^+}_{AB}\right]\nonumber\\
        =&\operatorname{Tr}\left[W'^\dagger\widetilde{\rho}_{AB}W'\ket{0}\!\bra{0}_{AB}\right]
        \label{eq:ED_loss_func}
    \end{align}

Thus, the depth-2 circuit $W'$ is also employed. We remark that in entanglement detection we may not need the cost function to converge; the algorithm can halt when the cost function $C^{ED}\ge1$ and the output is a positive detection result.

\subsection{Parameterized quantum circuit}

As discussed in Section \ref{sec:theoframe}, we adopt the bilocal PQC $U(\bm{\theta}_1)\otimes V(\bm{\theta_2})$ for variational Schmidt decomposition and the one-side PQC $U(\bm{\theta}_1)\otimes I$ for logarithm negativity estimation. Since our algorithm makes no assumption about input states, we recommend the class of hardware-efficient PQCs \cite{bharti2021noisy} as the tunable unitary to further enhance the experimental realizability of our algorithms. In the experiments, the PQCs consist of single qubit rotations and CNOT gates or controlled-Z gates as the entangled gates. We refer the readers to Sec. \ref{sec:experiments} for more details of the circuit implementation in each experiment.

\subsection{Optimization}

Both gradient-based \cite{Mitarai2018quantum,Stokes_2020} and gradient-free \cite{parrish2019jacobi,Nakanishi2020Sequential,Ostaszewski_2021} optimization methods can be applied to our variational quantum algorithms. For gradient-based optimization, the analytic gradients of $C^{SD}$, $C^{LN}$ and $C^{ED}$ are supported by the parameter shift rule \cite{Schuld_2019}, which provides an unbiased estimate compared to the finite difference method \cite{Benedetti_2019}. We have adopted various optimization methods in both the numerical simulation and the real-device implementation of our algorithms (see Section \ref{sec:experiments} for more details about the optimization adopted in our experiments).

Here we also discuss the possible solutions to the gradient vanishing issue for our algorithms. According to Ref. \cite{McClean_2018}, our algorithms could suffer from `barren plateaus' throughout training. Additionally, this gradient vanishing issue is unlikely to be solved by a local cost function \cite{Cerezo_2021} due to the global connectivity brought by subcircuit $W^{-1}$. On the other hand, proposed strategies including layerwise learning \cite{skolik2021layerwise}, parameter correlations \cite{Volkoff_2021}, and a quantum convolutional neural network ansatz \cite{pesah2020absence} have been shown to
be helpful in certain training tasks. We leave the adaptation of these and other training techniques to our algorithms for future study.

\subsection{Summary of algorithms and comparisons}

Before moving to the experiments, we summarize our algorithms and explain improvements over known methods. From the previous discussion, one can observe that our algorithms share a similar framework: the target state is operated by local parameterized quantum circuits, and the computational basis measurement after a depth-2 subcircuit estimates the fidelity with some elaborated entangled state as the optimization function. The VQA approach provides practicality on near-term quantum devices, and our framework guarantees their efficiency due to the constant extra consumption of computational resources. Exploiting symmetry, our framework can further reduce the number of parameters in the optimization of the VQA. This idea of ``computing the state overlap with some entangled state after operating local unitaries'' lies at the heart of our algorithms' framework, and turns out to reflect entanglement properties for near-term devices. As the algorithms advance over known variational approaches, our framework technically and conceptually contributes to entanglement analysis of near-term quantum devices, and is worth further study, development, and extension, \eg, in multipartite cases.

Here we make a detailed comparison between our algorithms and known methods. For the Schmidt decomposition algorithm, methods proposed in Refs. \cite{larose2019variational,bravo2020quantum} are also based on the hybrid quantum-classical approach. Essentially, all three methods perform local unitaries on each party, and should diagonalize the marginal states after optimization over unitaries. In this sense, it is reasonable to assume that the three methods need the same expressibility \cite{sim2019expressibility} for the ansatz.

Compared to Ref \cite{larose2019variational}, our algorithm demands no extra qubit, while in their approach (with a global cost function), two copies of the state are consumed in each optimization iteration. Thus, the resource of state preparation as well as the quantum register is doubled. The method in Ref. \cite{bravo2020quantum} seems to require a slightly shallower circuit than our algorithm due to the 2-layer depth subcircuit $W^\dagger$. Nevertheless, our algorithm can detect the case where the input state is polluted into a separable but classically correlated state. With a similar argument as in the proof of Lemma \ref{lemma:containingRstate}, one can show that for a separable state $\max C^{SD}\le\max p_i^2$. So, if, for example, the input pure state $0.6\ket{00}+0.8\ket{11}$ decoherence to mixed state $0.36\proj{00}+0.64\proj{11}$, our algorithm will detect this decoheres, while the method in Ref. \cite{bravo2020quantum} will not be able to distinguish them.

For the entanglement detection algorithm, recent work in Ref. \cite{wang2022detecting} realized various positive (but not completely positive) map criteria, including the reduction criterion in the variational approaches. Their method employs the quasiprobability sampling, which suffers from an exponential sampling overhead in the number of qubits. Although our algorithm, as a separability criterion, is no stronger than what Wang \etal \cite{wang2022detecting} implemented, we estimate only one expectation value in each optimization iteration that saves us from exponential resource consumption for large quantum systems.

\section{Numerical simulation and implementation on quantum device}\label{sec:experiments}

In this section, we show the effectiveness of variational quantum algorithms in Schmidt decomposition and logarithm negativity estimation through numerical simulation and implementation on quantum devices. The simulation experiments are operated on the \href{https://qml.baidu.com}{\textit{Paddle quantum}} platform and the \href{https://quantum-hub.baidu.com/}{\textit{Quantum Leaf}} platform. We also realize our algorithm on the superconducting quantum device of Institute of Physics, Chinese Academy of Sciences through the \href{https://quantum-hub.baidu.com/}{\textit{Quantum Leaf}} platform. 

\subsection{Simulating on circuits with different depths}\label{subsec:numerical}

Here, we conduct numerical experiments to decompose a bipartite quantum system with 8 qubits to investigate the impact of different numbers of circuit layers. In the beginning, we randomly generate an 8-qubit bipartite entangled state $\ket{\psi}_{AB}$, and both parties have 4 qubits. Then, ewe apply PQCs $U_A(\vec{\theta_A})$ and $V_B(\vec{\theta_B})$ on the two parties respectively. The structure of the PQCs used in this task is shown in Fig. \ref{fig:layer}. As discussed above we adopt Eq. \eqref{eq:C^SD_W} as the cost function. This simulation is performed on the \href{https://qml.baidu.com}{\textit{Paddle Quantum}} platform.

\begin{figure}[ht]
\centering
\includegraphics[width=0.4\textwidth]{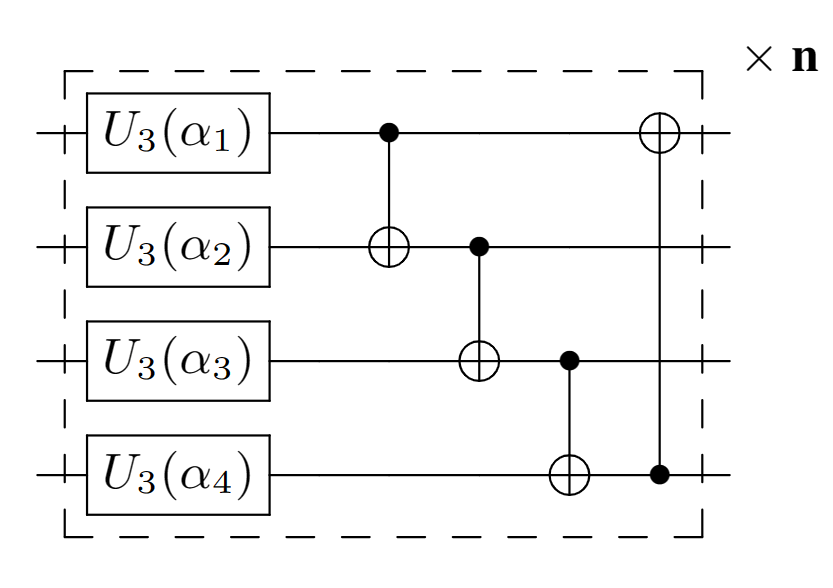}
\caption{Structure of one entangled layer. Here $U_3$ is the universal single-qubit template ($ZYZ$ rotation gate), $\alpha_j$ is the parameter of the gate, and \textbf{n} is the depth of PQC.}
\label{fig:layer}
\end{figure}

In this experiment, we set the depth of layers of PQCs to be $1,2,4,8$ , which means applying the structure in Fig. \ref{fig:layer} recurrently for depth times, to clarify the impact of the number of circuit layers. The results are shown in FIG. \ref{fig:DEPTH}. We use ADAM optimizer (a gradient-based optimizer) to maximize our cost function.

The simulation results are shown in Fig. \ref{fig:DEPTH}. The error (vertical axis) is defined as the squared $L2$ distance between the real and estimated Schmidt coefficient vectors, \ie, $\sum_i |c_j-c_j'|^2$, where the $c_j'$ are the estimated Schmidt coefficients and the $c_j$ are the actual values. It is obvious that circuits with larger depth have better performance in accuracy, satisfying the intuition that deeper circuits have better expressibility \cite{sim2019expressibility}.

\begin{figure}[ht]
\centering
\includegraphics[width=0.5\textwidth]{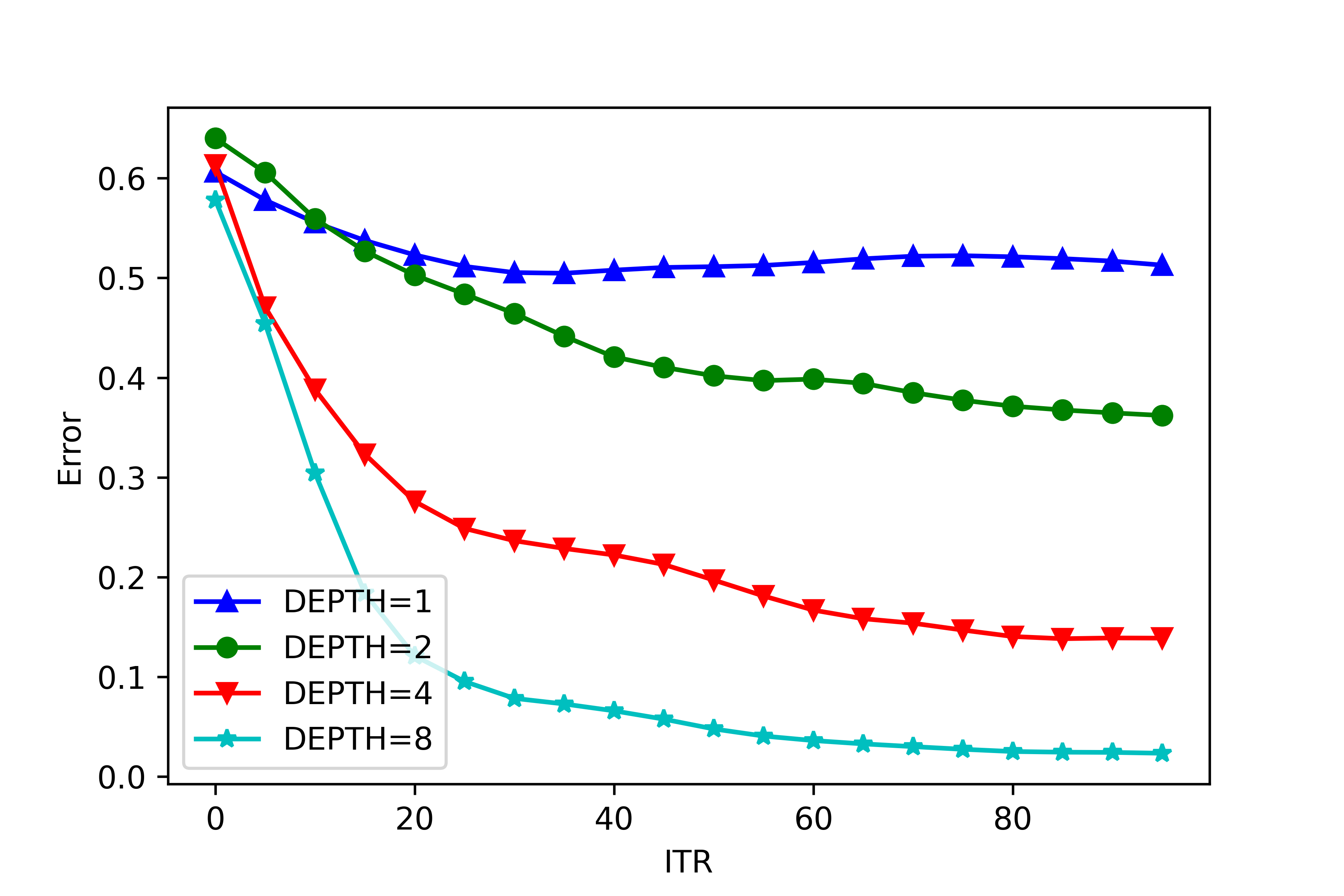}
\caption{Simulating 8-qubit Schmidt decomposition on PQCs with depth=1, 2, 4, and 8.}
\label{fig:DEPTH}
\end{figure}

\subsection{Decomposing the noisy state}\label{subsec:Decompose noisy state}
Next, we investigate the performance of our algorithm in decomposing quantum states affected by noisy channels. We take the amplitude damping channel $N^{amp}_p(\rho)$ and depolarizing channel $N^{depl}_p(\rho)$ into consideration, which are defined as
\begin{align*}
    N^{amp}_p(\rho) := & E_0\rho E_0^\dagger + E_1\rho E_1^\dagger,\\
    N^{depl}_p(\rho) := & (1-p)\rho+p\tr(\rho)\frac{I}{2},
\end{align*}
respectively, and

\begin{align*}
	E_0=\begin{bmatrix}1 & 0 \\
		0 & \sqrt{1-p}\end{bmatrix},E_1=\begin{bmatrix}0 & \sqrt{p} \\
		0 & 0\end{bmatrix},
\end{align*}
with the noise level $p\in [0,1)$. In this simulation, we decompose a bipartite system with one qubit in each party, and the input state is
\begin{equation}
\begin{aligned}
    \rho_{AB} &= \proj{\psi}_{AB}\\
    &= 
    \begin{pmatrix}
     0.455 &  0.459 &  0.136 & -0.137 \\
     0.459 &  0.463 &  0.137 & -0.138\\
     0.136 &  0.137 &  0.041 & -0.041\\
    -0.137 & -0.138 & -0.041 &  0.041
    \end{pmatrix}.
\end{aligned}
\label{eq:bipartite state}
\end{equation}

The Schmidt coefficients of the state are $c_1=0.958,c_2=0.286$. We employ noisy channels to each party before applying PQCs $U(\mathbf{\theta_1})$ and $V(\mathbf{\theta_2})$ which are single-qubit universal ansatz. This simulation is performed on the \href{https://qml.baidu.com}{\textit{Paddle quantum}} platform.
\begin{figure}[ht]
\centering
\includegraphics[width=0.5\textwidth]{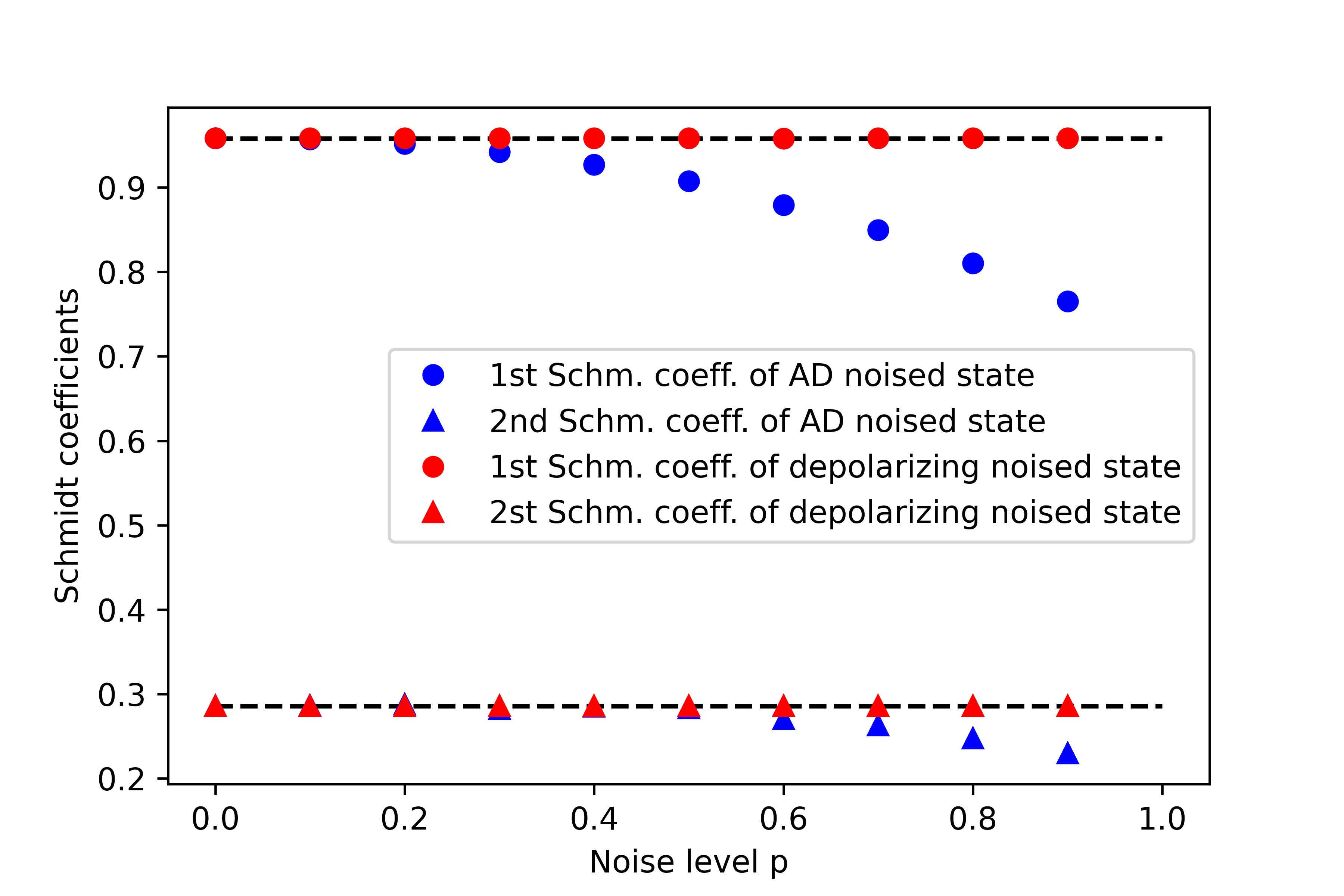}
\caption{Schmidt coefficients of quantum states affected by noisy states. The blue circles and triangles are the Schmidt coefficients of the amplitude damping (AD) noised state. The red circles and triangles are the Schmidt coefficients of the depolarizing noised state. The black dashed line is the theoretical Schmidt coefficients of the state without noise.}
\label{fig:noisy_state}
\end{figure}

Fig. \ref{fig:noisy_state} displays the Schmidt coefficients of the states affected by both noisy channels with different noise levels $p$. For the depolarizing channel, the estimated Schmidt coefficients equal the theoretical values (black dashed line) all the time; for the amplitude damping channel, when the noise level is low (smaller than 0.5), the estimated Schmidt coefficients are close to the theoretical values but diverge as the noise level increases. We can conclude that our algorithm is robust to depolarizing channels, and can achieve relatively accurate results if the state is noised by an amplitude damping channel (low noise level).

\subsection{Implementations on a superconducting quantum processor}\label{subsec:Quantum device}

We also apply our algorithm to decompose a 2-qubit bipartite quantum system on quantum devices and compare the results with the values achieved from the simulation. The experiments are performed on the \href{https://quantum-hub.baidu.com/}{\textit{Quantum Leaf}} platform, loading the quantum device from the Institution of Physics, Chinese Academy of Science (IoP CAS). This quantum device contains 10 direct coupling transmon qubits, whose topology is neighbor coupling one-dimensional chainlike.

The input state $\ket{\psi}_{AB}$ we use here is exactly the same as that in Sec. \ref{subsec:Decompose noisy state}. On the quantum device, the input state could be prepared by applying two $Ry$ gates on the first and second qubits with parameters $0.58$ and $1.58$, respectively, followed by a Controlled-$Z$ gate. Next, we operate the parameterized quantum circuit on qubits and use sequential minimal optimization \cite{Nakanishi2020Sequential} (gradient-free) to optimize the parameters until the cost converges to its maximum. We repeat 20 independent experiments with the same input states and randomly initialized parameters. 

\begin{figure}[ht]
\centering
\includegraphics[width=0.45\textwidth]{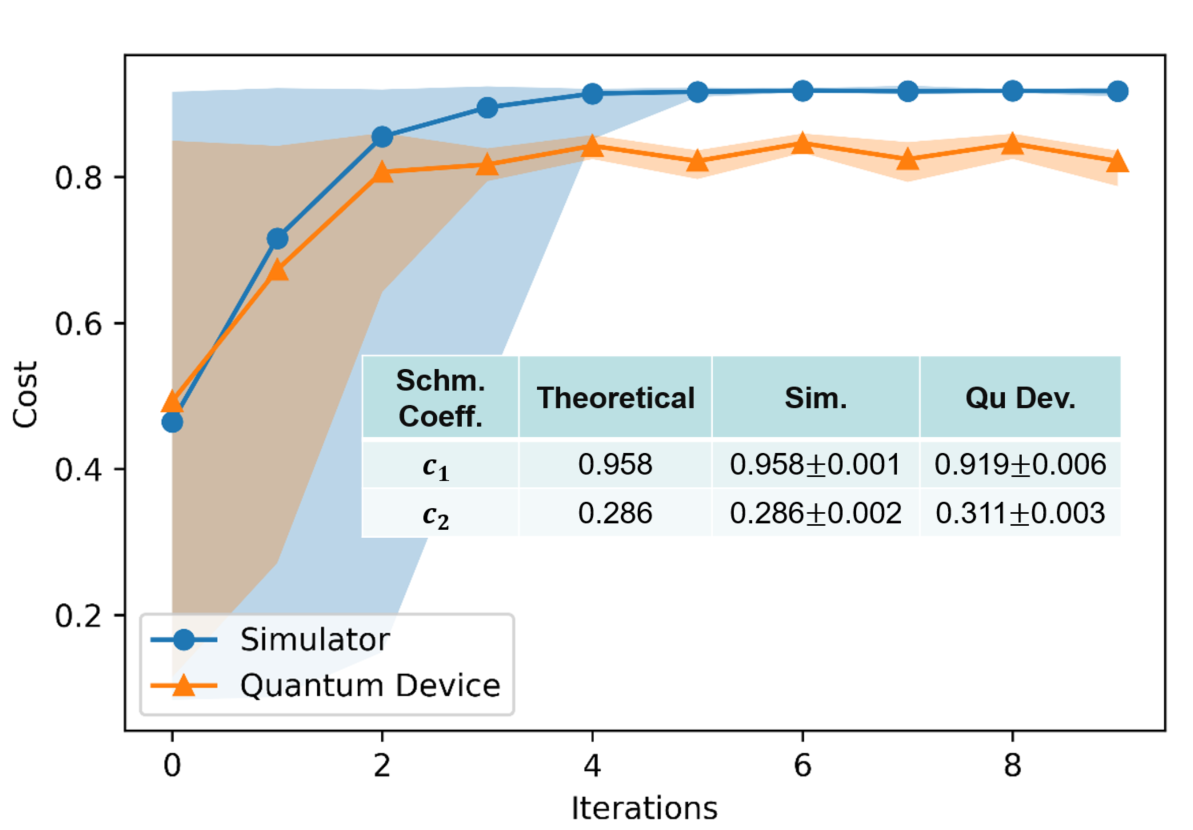}
\caption{Learning curves with respect to the simulator and quantum device. The shadowed area spans from the minimal to maximal value of 20 experiments. Theoretically, the Schmidt coefficients are $0.958$ and $0.286$, respectively. The simulator-estimated coefficients are $0.958 \pm 0.001$ and $0.286 \pm 0.002$, while the results from the CAS quantum device are $0.919 \pm 0.006$ and $0.311 \pm 0.003$.}
\label{fig:dev_vs_sim}
\end{figure}

The experimental results are displayed in Fig. \ref{fig:dev_vs_sim}. As demonstrated, the cost raises dramatically in the first two iterations, and the simulator converges to a higher value than that of the quantum device, which may be caused by the quantum device noise. Correspondingly, the achieved Schmidt values by the simulator ($c_1=0.958 \pm 0.001$ and $c_2=0.286 \pm 0.002$) are closer to the theoretical values ($c_1=0.958$ and $c_2=0.286$) than those of the quantum device ($c_1=0.919\pm0.006$ and $c_2=0.311\pm0.003$). 

\subsection{Estimating the logarithm negativity}\label{subsec:Logarithm negativity}
From Corollary \ref{cor:norm}, we could derive the logarithm negativity estimation, which quantifies the entanglement of the input state $\ket{\psi}_{AB}$. Here, we investigate the accuracy of our method in estimating the logarithm negativity at different input state ranks by comparing it with theoretical values. 

In this simulation, $\ket{\psi}_{AB}$ is bipartite with 3 qubits in each party. In order to make the results comparable, the amplitude of the input state is designed to be identity, \ie, $\ket{\psi}=\frac{1}{\sqrt{r}}\sum_{i=1}^r \ket{i}$, where $r$ denotes the Schmidt rank ranging from 1 to 8. The circuit for the state preparation is displayed in Appendix \ref{sec:SM_beta}. Then, we apply the parameterized quantum circuit, which is shown in Fig. \ref{fig:LG_cir}, followed by the subcircuit $W^{-1}$ as in Fig. \ref{fig:diagram}. The sequential minimal optimization \cite{Nakanishi2020Sequential} is utilized to maximize the cost function, which is set as Eq. \eqref{eq:LN_loss_func}, until it converges to its maximum value. At this stage, the logarithm negativity of the input state $\ket{\psi}$ can be calculated by Eq. \eqref{eq:LN_cal}. We repeat the calculation for logarithm negativity 10 times for the input state with different Schmidt ranks, \ie, $r=1,2,3,4,5,6,7,8$, and compare the calculated values with the theoretical values. This simulation is performed on the \href{https://quantum-hub.baidu.com/}{\textit{Quantum Leaf}} platform. The results are shown in Fig. \ref{fig:rank}. From the results, we can tell that the estimated values from our method are close to the ideal values even for input states with a high Schmidt rank. The difference between the theoretical and estimated values comes from the statistical error due to finite sampling. In this numerical simulation, the results are calculated by sampling each circuit execution 1024 times. It should be noted that if a shallow circuit is adopted, the experimental results may exhibit a degradation in performance.

\begin{figure}[ht]
\centering
\includegraphics[width=0.4\textwidth]{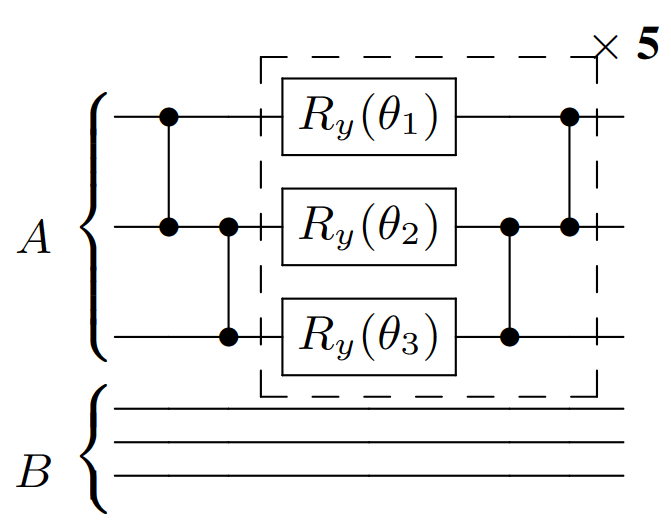}
\caption{Parameterized circuit ansatz for logarithm negativity estimation. The circuits in the dashed box should be repeated five times.}
\label{fig:LG_cir}
\end{figure}

\begin{figure}[ht]
\centering
\includegraphics[width=0.45\textwidth]{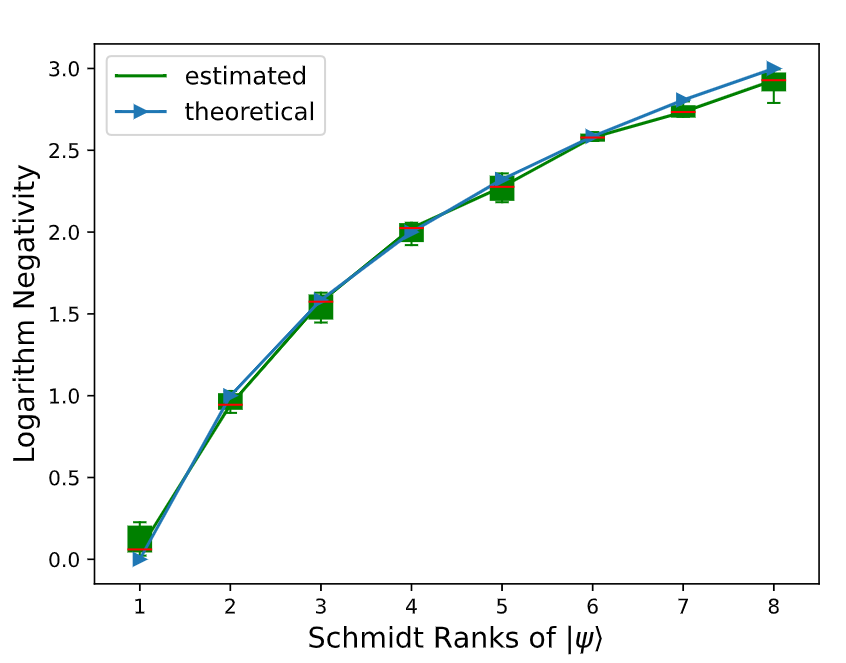}
\caption{Logarithm negativity estimation with respect to different Schmidt ranks. The red bar is the median logarithm negativity calculation repeated 10 times, and the green rectangles and bar stand for the standard deviation and error, respectively. This experiment includes the no-entanglement (rank $r=1$) case and the maximally entangled ($r=8$) case.}
\label{fig:rank}
\end{figure}

\section{Conclusion and outlook}\label{sec:conclusion}

We proposed variational quantum algorithms to realize Schmidt decomposition and estimated logarithm negativity for pure states, and detect entanglement for general mixed states on near-term quantum devices. Compared to previous methods, our Schmidt decomposition algorithms are more efficient in terms of resource consumption or more stable for polluted inputs. For general mixed states, a new method to detect entanglement on near-term quantum devices was derived, which efficiently detects entanglement for specific families of states and detects distillability in general. We have shown the validity and practicality of our methods via both numerical simulations and experimental implementations. Numerical simulations show that the variational Schmidt decomposition algorithm is resilient to amplitude damping and depolarizing noises. Experiments on superconducting quantum devices show relatively accurate results, illustrating that our algorithms are executable and valid on near-term quantum devices. Our variational quantum algorithms notably provide efficient and practical tools for entanglement quantification and analysis in the NISQ era.

Beyond quantifying entanglement for bipartite pure states, our algorithms could have a wide range of applications and extensions. Direct applications of our work include quantum data compression. For example, if the two parties of the state have distinct dimensions (or numbers of qubits), it is always possible for the larger party to reduce its scale by applying
the unitary transformation obtained through training. Moreover, it would be of independent interest to further explore the power of local PQC, which is a key ingredient in our method. One more extension of our framework is to consider the multiparty case. If the Schmidt decomposition of a multiparty entangled state exists \cite{das2018necessary, pati2000existence}, our method can be adapted to decompose these states. We leave the adaptation of our techniques to investigate multiparty pure or mixed states for future study, as well as its implications in many-body quantum systems that exist in today's various quantum computation platforms \cite{Bruzewicz_2019,Kjaergaard2020}.

\section*{Code availability}
The simulation and experimental codes are available from Github\cite{code_website}.

\section*{Acknowledgements.}

We thank Zelin Meng for helpful discussions. We also appreciate valuable help from the developers of the Quantum Leaf platform as well as the engineers of the IoP-CAS superconducting device. Part of this work was done when R. C., B. Z., and X. W. were at Baidu Research.

\bibliography{smallbib}

\onecolumngrid
\vspace{2cm}
\begin{center}
{\textbf{\large Appendix of Variational Quantum Algorithms for Schmidt Decomposition}}
\end{center}

\renewcommand{\theequation}{S\arabic{equation}}
\renewcommand{\theproposition}{S\arabic{proposition}}
\renewcommand{\thelemma}{S\arabic{lemma}}

\setcounter{equation}{0}
\setcounter{table}{0}
\setcounter{section}{0}
\setcounter{proposition}{0}
\setcounter{lemma}{0}

\section{Selection of parameters in subcircuit \texorpdfstring{$W_A$}{}}\label{sec:SM_alpha}

Here we show that how to choose parameters $\alpha_j$ for $j\in[1,n]$ in the circuit $W_A$ defined in Section \ref{subsec:costfunction} when $n\ge2$. 

We first choose three real numbers $\beta_n$, $\beta_{n-1}$, and $\Delta\beta$ satisfying $1<\beta_n<\beta_{n-1}$ and $\Delta\beta>0$. If $n>2$, we recursively set $\beta_j=\Pi_{k=j+1}^{n}\beta_k+\Delta\beta$ for $j=n-2,\dots,1$. Then we let $\gamma_j=\beta_j/(\beta_j+1)$ and let $\alpha_j=2\arccos\sqrt{\gamma_j}$. We remark that the parameters are evaluated on the classical computer once for each algorithm run, with running time $O(n)$.

To show that $W_A$ prepares the state $\ket{\Phi}_{A}=\sum_jp_j\ket{j}_A$ such that $p_j$ strictly decreases, let's look at the coefficients $\{p_j\}$ determined by $\{\alpha_j\}$. The output state can be written as 
\begin{align*}
    \ket{\Psi}_{A}=&\bigotimes_{j=1}^n\left(\cos\frac{\alpha_j}{2}\ket{0}_j+\sin\frac{\alpha_j}{2}\ket{1}_j\right)\nonumber\\
    =&\bigotimes_{j=1}^n\left(\sqrt{\gamma_j}\ket{0}_j+\sqrt{1-\gamma_j}\ket{1}_j\right).
\end{align*}

Note that $\beta_j>1$ implies $\gamma_j>1-\gamma_j$. Let the binary representation of $j$ be $j=(j_1j_2\dots j_n)_2$ where $j_l,l\in[1,n]$ are binary bits. Then the coefficient $p_j$ is 
\begin{align*}
    p_j=\Pi_{l=1}^n\sqrt{\gamma_l^{1-j_l}(1-\gamma_l)^{j_l}}.
\end{align*}

(For example, when $n=4$ and $j=9=(1001)_2$, we have  $p_j=p_9=\sqrt{\gamma_1^0\gamma_2^1\gamma_3^1\gamma_4^0(1-\gamma_1)^1(1-\gamma_2)^0(1-\gamma_3)^0(1-\gamma_4)^1}=\sqrt{(1-\gamma_1)\gamma_2\gamma_3(1-\gamma_4)}$.)

For any $j<k$, denote by $n_0$ the highest bit such that $j_{n_0}$ and $k_{n_0}$ are different (hence $j_{n_0}=0,k_{n_0}=1$ and $j_l=k_{l}$ for $l<{n_0}$). Consider the index $j_+=(j_1\dots j_{n_0}11\dots1)_2$ and $k_-=(k_1\dots k_{n_0}00\dots0)_2$, so $j\le j_+$ and $k\ge k_-$. Also we have that
\begin{align*}
    \frac{p_j}{p_{j_+}}=\Pi_{l=n_0,j_l\neq1}^n\sqrt{\frac{\gamma_l}{1-\gamma_l}}\ge1, \frac{p_k}{p_{k_-}}=\Pi_{l=n_0,k_l\neq0}^n\sqrt{\frac{1-\gamma_l}{\gamma_l}}\le1,
\end{align*}
so $p_j\ge p_{j_+}$ and $p_k\le p_{k_-}$. Finally, we compare $p_{j_+}$ and $p_{k_-}$:
\begin{align*}
    \frac{p_{j_+}}{p_{k_-}}=\sqrt{\frac{\gamma_{n_0}}{1-\gamma_{n_0}}\Pi_{l=n_0+1}^n\frac{1-\gamma_l}{\gamma_l}}.
\end{align*}
By the definition of $\beta$ we have that $\beta_{n_0}>\Pi_{l=n_0+1}^n\beta_l$. Together with $\beta_j=\gamma_j/(1-\gamma_j)$ it holds that
\begin{align*}
    \frac{p_{j_+}}{p_{k_-}}=\sqrt{\frac{\gamma_{n_0}}{1-\gamma_{n_0}}\Pi_{l=n_0+1}^n\frac{1-\gamma_l}{\gamma_l}}=\sqrt{\beta_{n_0}\Pi_{l=n_0+1}^n\frac{1}{\beta_l}}>1.
\end{align*}
So $p_j\ge p_{j_+}> p_{k_-}\ge p_k$, guaranteeing a strictly decreasing sequence $\{p_j\}$.

\section{Circuits for preparing state with different Schmidt ranks}\label{sec:SM_beta}

Here we display the circuits to generate a 6-qubit bipartite entangled state with desired Schmidt rank. The circuits are shown below. It is notable that the first three qubits belongs to one party and left three qubits belongs to the other. For example, if we wish to prepare a input state with Schmidt rank equals to 5, we should firstly apply the circuit in Fig. \ref{fig:rank=5}, then operate the second circuit in Fig. \ref{fig:rank circuit}.

\begin{figure}[htbp]\label{fig:rank state generator}
\subfigure[Schmidt rank=1.]{
\includegraphics[width=4cm]{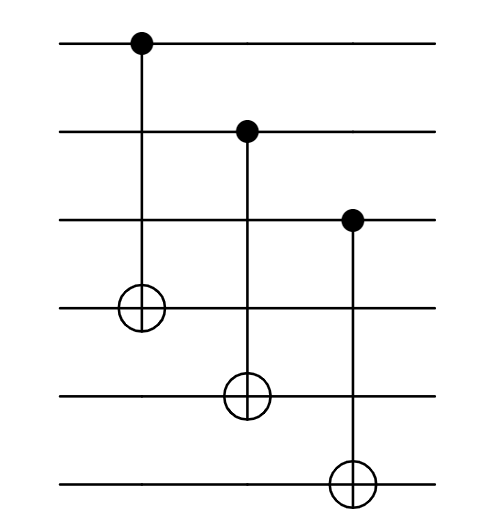}}
\subfigure[Schmidt rank=2.]{
\includegraphics[width=4cm]{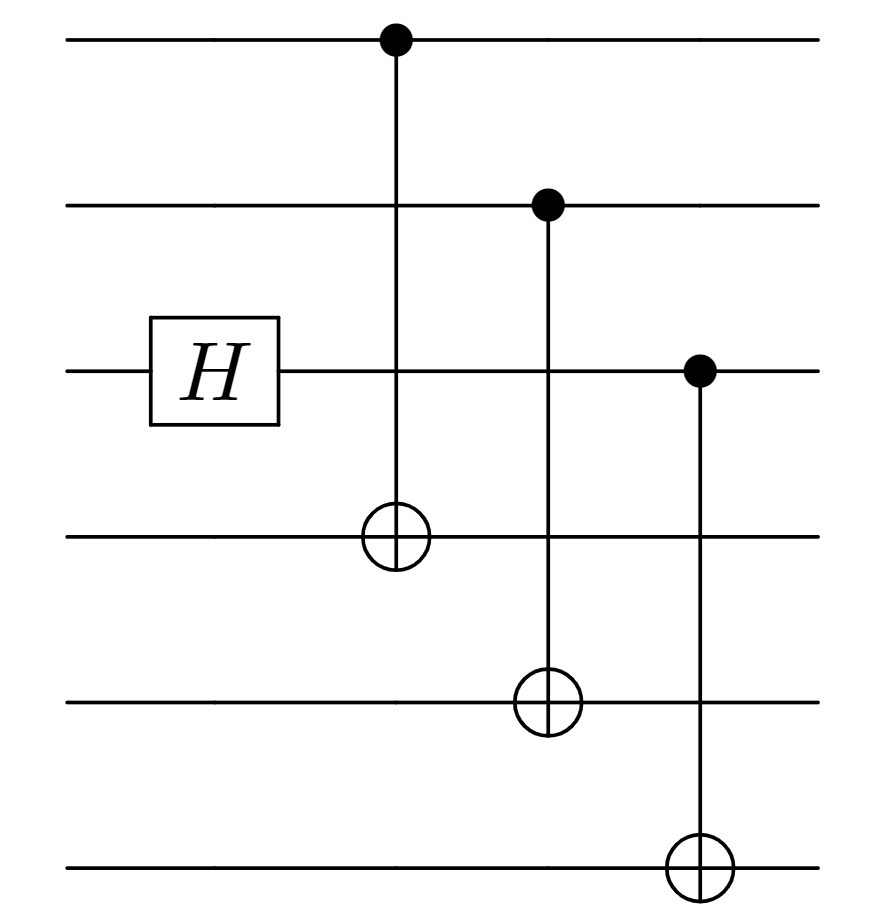}}
\subfigure[Schmidt rank=3, parameter $\alpha = 2\arcsin(\sqrt{1/3})\approx 1.230959$.]{
\includegraphics[width=7cm]{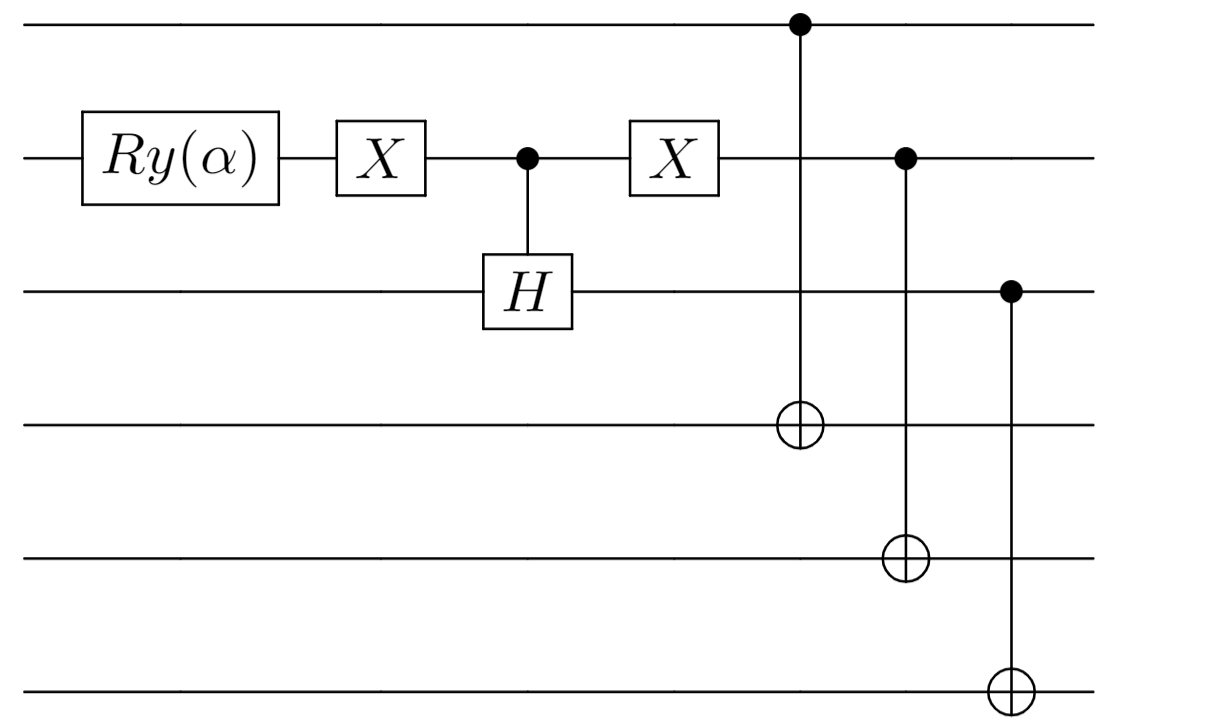}}
\subfigure[Schmidt rank=4.]{
\includegraphics[width=4cm]{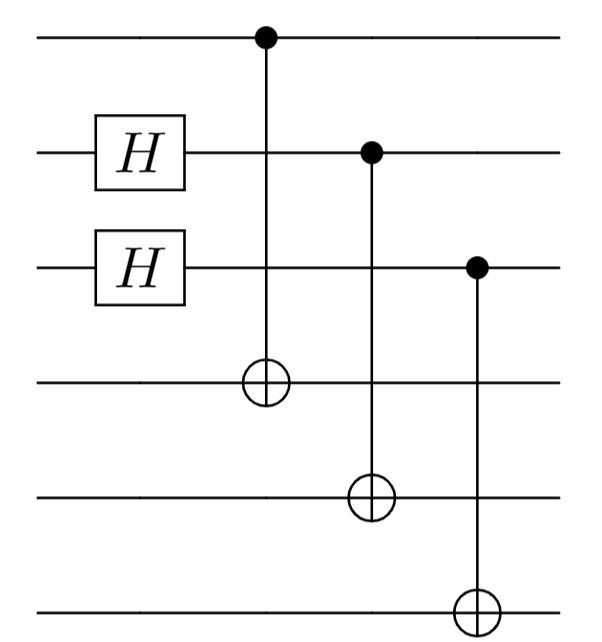}}
\subfigure[Schmidt rank=5, parameter $\alpha =2\arcsin(\sqrt{1/5})\approx 0.927295$.]{
\label{fig:rank=5}
\includegraphics[width=7.5cm]{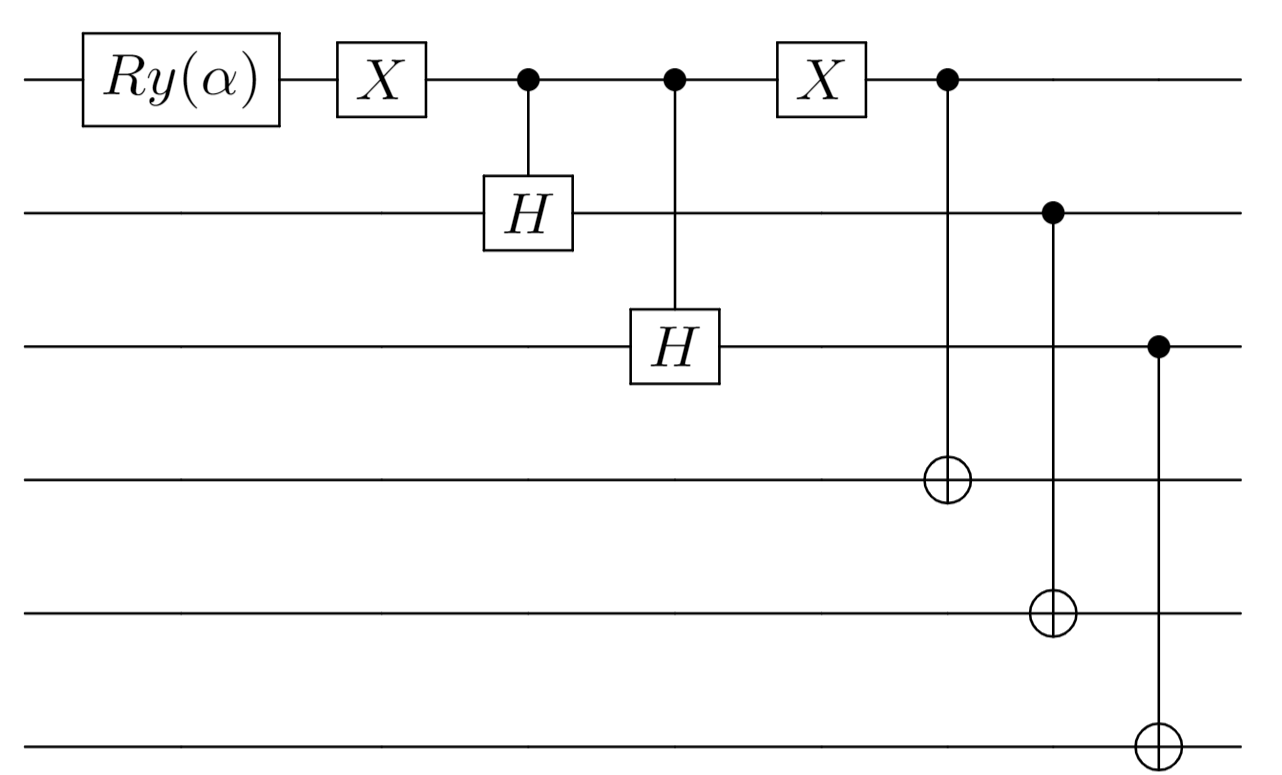}}
\subfigure[Schmidt rank=6, parameter $\alpha = 2\arcsin(\sqrt{1/3})\approx 1.230959$.]{
\includegraphics[width=6cm]{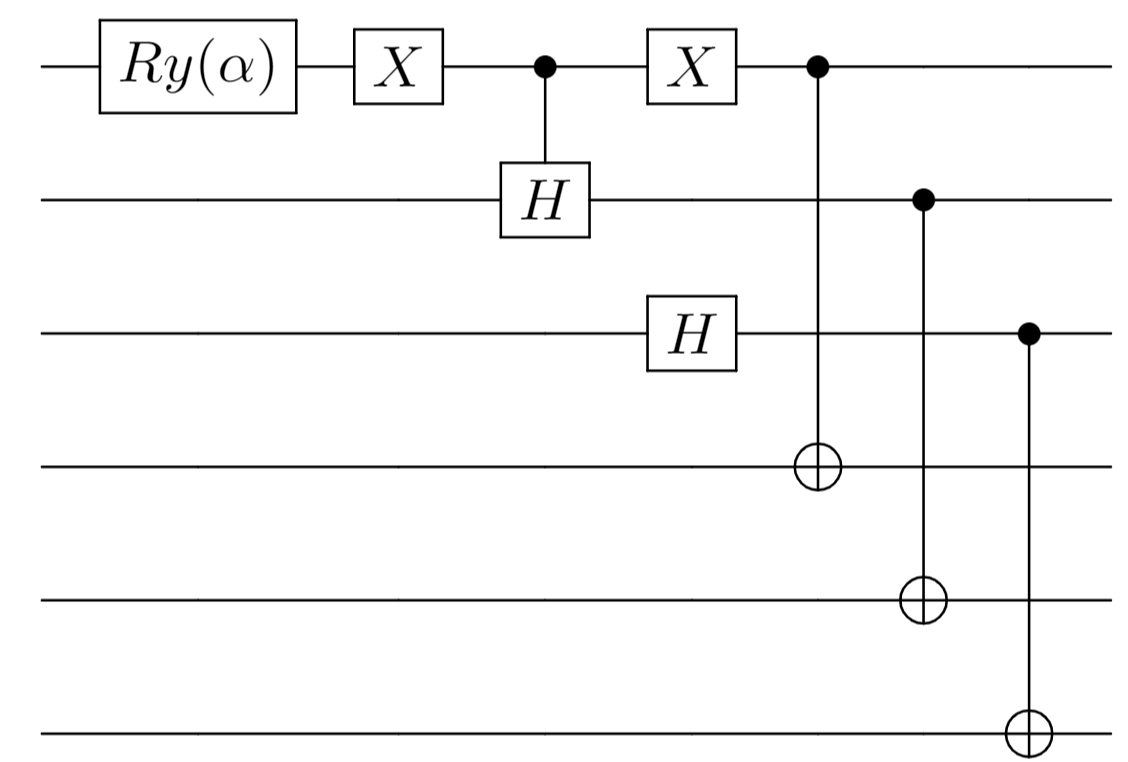}}
\subfigure[(approximate) Schmidt rank=7, parameter $\alpha_1=1.4$, $\alpha_2=\pi/4$, $\alpha_2=-\pi/4$.]{
\includegraphics[width=7cm]{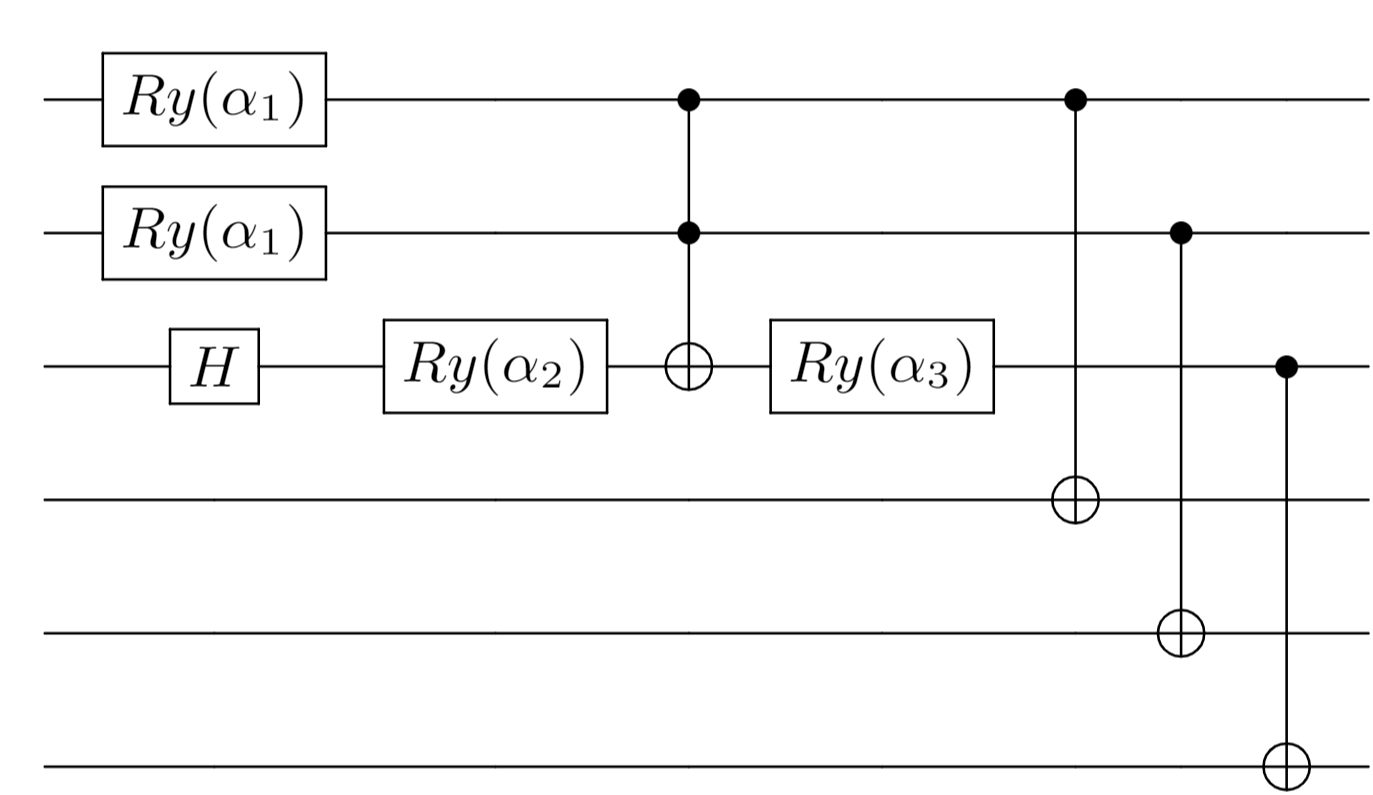}}
\subfigure[Schmidt rank=8.]{
\includegraphics[width=5cm]{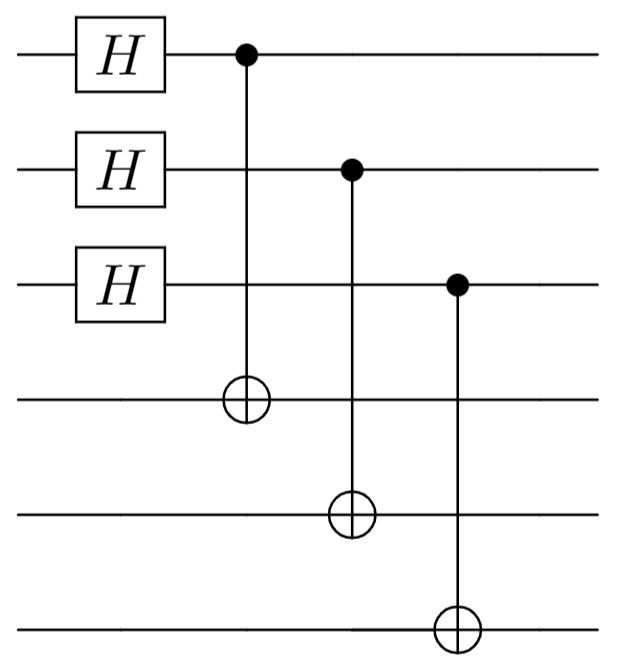}}
\subfigure[Circuit for second step to prepare input state. The parameters are random values.]{
\label{fig:rank circuit}
\includegraphics[width=5cm]{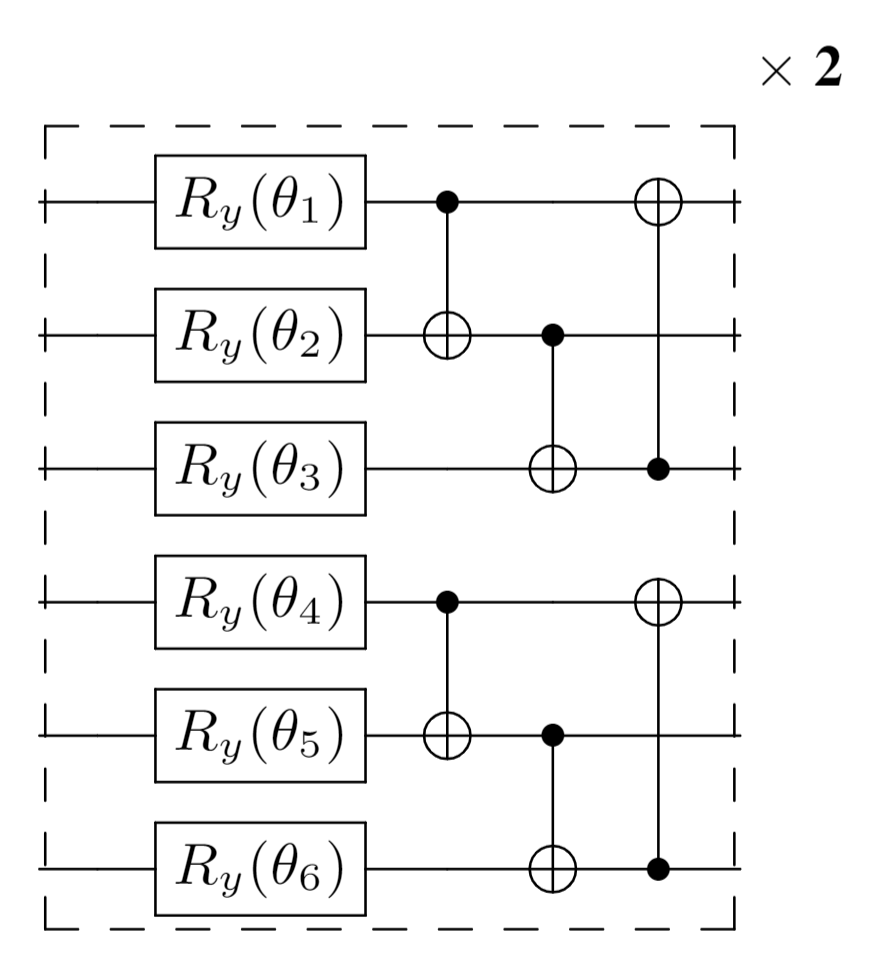}}
\caption{Quantum circuits for preparing input state with different Schmidt ranks. (a)$\sim$(h) are the first steps to prepare the input state, (i) is the second step.}
\end{figure}

\section{Algorithm boxes}\label{sec:SM_algo}
\renewcommand{\algorithmicrequire}{\textbf{Input:}}
\renewcommand{\algorithmicensure}{\textbf{Output:}}
\vspace{0cm}
\begin{algorithm}[H] 
\caption{Variational Quantum Algorithm for Schmidt Decomposition}
\label{alg:vqasd}
\begin{algorithmic} 
\REQUIRE bipartite $d$-dimension pure state $\rho_{AB}$, parameterized quantum circuit (PQC) $U(\bm\theta_1)$ and $V(\bm\theta_2)$, inverse of state preparing circuit $W^{-1}$ (Fig. \ref{fig:diagram}), number of iterations ITR;
\ENSURE Schmidt coefficients $\{c_j\}$ and Schmidt bases $\{u_j\}, \{v_j\}$.

\STATE Initialize parameters $\bm\theta$.

\STATE \color{blue}\texttt{\# Training iteration}\color{black}

\FOR{itr $=1,\ldots,$ ITR}
 
\STATE Apply $U(\bm\theta_1)$ to party $A$ and $V(\bm\theta_2)$ to party $B$ in state $\rho_{AB}$, and the state $\rho_{AB}$ becomes $\Tilde{\rho}_{AB}$.

\STATE Apply the inverse of state preparing circuit $W^{\dagger}$ to the state $\Tilde{\rho}_{AB}$.

\STATE Compute the cost function $C^{SD}=\text{Pr}[{\bm b}={\bm 0}]=\text{Tr}[W^{\dagger}\Tilde{\rho}_{AB}W\ket{0}\!\bra{0}_{AB}]$.

\STATE Maximize the cost function $C^{SD}$ and update parameters $\bm\theta$.

\ENDFOR

\STATE \color{blue}\texttt{\# Coefficient readout}\color{black}

\STATE Denote the trained PQCs as $U(\bm\theta^{opt})$ and $V(\bm\theta^{opt})$.

\STATE Apply the trained PQCs $U(\bm\theta^{opt})$ to the party $A$, and $V(\bm\theta^{opt})$ to the party $B$ in target state $\rho_{AB}$, then achieve $\Tilde{\rho}'_{AB}$.

\STATE Measure the reduced state $\Tilde{\rho}'_{A}$ on the computational basis $\{\ket{j}_A\}$. The $j$-th Schmidt coefficient is calculated by $c_j=\sqrt{\Pr[\bm b=j]}=\sqrt{\bra{j}\Tilde{\rho}'_{A}\ket{j}}$. 
The corresponding Schmidt bases are $u_j=U^{\dagger}(\bm\theta^{opt})\ket{j}_A$ and $v_j=V^{\dagger}(\bm\theta^{opt})\ket{j}_B$, respectively for $j\in [0, d-1]$.

\end{algorithmic}
\end{algorithm}

\begin{algorithm}[H] 
\caption{Variational Quantum Algorithm for Logarithm Negativity Estimation}
\label{alg:vqaln}
\begin{algorithmic} 
\REQUIRE bipartite $d$-dimension pure state $\rho_{AB}$, parameterized quantum circuit (PQC) $U(\bm\theta)$, inverse of state preparing circuit $W'^{-1}$ (Fig. \ref{fig:diagram}), number of iterations ITR;

\ENSURE logarithm negativity of the bipartite pure state $E_N(\rho_{AB})$.

\STATE Initialize parameters $\bm\theta$.

\STATE \color{blue}\texttt{\# Training iteration}\color{black}

\FOR{itr $=1,\ldots,$ ITR}

\STATE Apply $U(\bm\theta)$ to party $A$ in state $\rho_{AB}$, and the state  $\rho_{AB}$ becomes $\Tilde{\rho}_{AB}$.

\STATE Apply the inverse of state preparing circuit $W'^{\dagger}$ to the state $\Tilde{\rho}_{AB}$.

\STATE Calculate the loss function $C^{LN} =\text{Pr}[{\bm b}={\bm 0}]=\text{Tr}[W'^{\dagger}\Tilde{\rho}_{AB}W'\ket{0}\!\bra{0}_{AB}]$.

\STATE Maximize the loss function $C^{LN}$, and update parameters $\bm\theta$.

\ENDFOR

\STATE \color{blue}\texttt{\# Logarithm negativity readout}\color{black}

\STATE Denote the maximum loss function as $C^{LN}_{max}$, and the logarithm negativity is calculated by $E_N(\rho_{AB}) = \text{log}_2(C^{LN}_{max})+n$

\end{algorithmic}
\end{algorithm}

\begin{algorithm}[H] 
\caption{Variational Quantum Algorithm for Entanglement Detection}
\label{alg:vqaed}
\begin{algorithmic} 
\REQUIRE bipartite $d$-dimension quantum state $\rho_{AB}$, parameterized quantum circuit (PQC) $U(\bm\theta)$, inverse of state preparing circuit $W'^{-1}$ (Fig. \ref{fig:diagram}), number of iterations ITR;

\ENSURE entanglement detection of the bipartite quantum state.

\STATE Initialize parameters $\bm\theta$.

\STATE \color{blue}\texttt{\# Training iteration}\color{black}

\FOR{itr $=1,\ldots,$ ITR}

\STATE Apply $U(\bm\theta)$ to party $A$ in state $\rho_{AB}$, and the state  $\rho_{AB}$ becomes $\Tilde{\rho}_{AB}$.

\STATE Apply the inverse of state preparing circuit $W'^{\dagger}$ to the state $\Tilde{\rho}_{AB}$.

\STATE Calculate the loss function $C^{ED} =\text{Pr}[{\bm b}={\bm 0}]=\text{Tr}[W'^{\dagger}\Tilde{\rho}_{AB}W'\ket{0}\!\bra{0}_{AB}]$.

\STATE Maximize the loss function $C^{ED}$, and update parameters $\bm\theta$.

\STATE \textbf{If} loss function $C^{ED}>1/d$:

\STATE ~~Entanglement is detected.

\STATE ~~\textbf{break}

\ENDFOR

\STATE \textbf{If} loss function $C^{ED}>1/d$:

\STATE ~~Entanglement is detected.

\end{algorithmic}
\end{algorithm}

\section{Examples that a $\chi$-state violates the reduction criterion}

\label{sec:SM_apnoisy}
We show that certain AD-noisy Bell states (Bell state under one-sided amplitude damping noise) have $\chi<1/d=1/2$, but violates the reduction criterion. 

Let $E_0=\proj{0}+\sqrt{1-\gamma}\proj{1}$, $E_1=\sqrt{\gamma}\ket{0}\!\bra{1}$, $\mathcal{N}_{AP}(\cdot)=E_0\cdot E_0^\dagger+E_1\cdot E_1^\dagger$, then the AD-noisy Bell states are expressed by $\rho=\mathcal{N}_{AP}\otimes \operatorname{id}(\proj{\Phi^+_2})$.

Let the PQC $U=R_z(z_1/2)R_y(y/2)R_z(z_2/2)$, then direct calculation shows
$$
\chi(\rho)=\frac{1}{2}\max_{z_1,y,z_2}1+\sqrt{1-\gamma}\cos(z_1+z_2)+\cos(y)(1-\gamma+\sqrt{1-\gamma}\cos{(z_1+z_2)})=\frac{2+2\sqrt{1-\gamma}-\gamma}{4}.
$$
So $\chi(\rho)>1/2\iff\gamma<2\sqrt{2}-2\approx0.828$. On the other hand, all AD-noisy Bell states are entangled unless $\gamma=1$, as they all violate the PPT criterion. So $\chi$ fails to detect entanglement when $\gamma\in(2\sqrt{2-2},1)$. But the reduction criterion can detect entanglement for those states, as it is a necessary and sufficient condition in the $2\times2$ case. 

\end{document}